\documentclass[12pt,twoside]{article}

 \usepackage{float}
\usepackage{graphicx}
\usepackage{epstopdf}
\usepackage{graphicx}
\usepackage{epic}
\usepackage{multirow}
\usepackage{tikz}
\usepackage{threeparttable}
\usepackage{xcolor}
\usepackage{rotating}

\usetikzlibrary{arrows,shapes,chains}

\usepackage[a4paper]{geometry}
\makeatletter
\renewcommand\title[1]{\gdef\@title{\reset@font\Large\bfseries #1}}
\renewcommand\section{\@startsection {section}{1}{\z@}%
                                   {-3.5ex \@plus -1ex \@minus -.2ex}%
                                   {2.3ex \@plus.2ex}%
                                   {\normalfont\large\bfseries}}
\renewcommand\subsection{\@startsection{subsection}{2}{\z@}%
                                     {-3ex\@plus -1ex \@minus -.2ex}%
                                     {1.5ex \@plus .2ex}%
                                     {\normalfont\normalsize\bfseries}}
\renewcommand\subsubsection{\@startsection{subsubsection}{3}{\z@}%
                                     {-2.5ex\@plus -1ex \@minus -.2ex}%
                                     {1.5ex \@plus .2ex}%
                                     {\normalfont\normalsize\bfseries}}

\def\@runningauthor{}\newcommand{\runningauthor}[1]{\def\runningauthor{#1}}
\def\@runningtitle{}\newcommand{\runningtitle}[1]{\def\runningtitle{#1}}

\renewcommand{\ps@plain}{%
\renewcommand{\@evenhead}{\footnotesize\scshape \hfill\runningauthor\hfill}
\renewcommand{\@oddhead}{\footnotesize\scshape \hfill\runningtitle\hfill}}

\newcommand{\F}{\mathbb{F}}

\newcommand{\Torus}{\mathcal{T}}
\newcommand{\TorusZero}{\mathcal{T}_0}
\newcommand{\Norm}{\operatorname{N}}
\newcommand{\ellT}{\ell_{\Torus}}
\newcommand{\wt}{\operatorname{wt}}

\newcommand{\be}{\begin{eqnarray}}
\newcommand{\ee}{\end{eqnarray}}
\newcommand{\nn}{{\nonumber}}

\newcommand{\ra}{\rightarrow}

\g@addto@macro\bfseries{\boldmath}

\makeatother

\usepackage{amsthm,amsmath,amssymb}
\usepackage{cite}

\usepackage[colorlinks=true,citecolor=black,linkcolor=black,urlcolor=blue]{hyperref}

\theoremstyle{plain}
\newtheorem{theorem}{Theorem}[section]

\newtheorem{lem}[theorem]{Lemma}
\newtheorem{cor}[theorem]{Corollary}
\newtheorem{prop}[theorem]{Proposition}

\theoremstyle{definition}
\newtheorem{definition}[theorem]{Definition}
\newtheorem{example}[theorem]{Example}

\newtheorem{algorithm}{Algorithm}

\theoremstyle{remark}
\newtheorem{remark}[theorem]{Remark}

\runningauthor{}

\date{}

\begin{document}

\title{Norm-One Torus Decompositions and Decoding of Gashkov-Sidel'nikov Codes\thanks{Minjia Shi was supported by the National Natural Science Foundation of China under Grant 12471490. 
Shitao Li was supported by the National Natural Science Foundation of China
under Grant 12526612.
Tor Helleseth was supported by the Research Council of Norway under Grant 247742/O70. Ferruh \"{O}zbudak is supported by	T\"{U}B$\dot{\mathrm{I}}$TAK under Grant 223N065.}
}
\author{Minjia Shi\thanks{Minjia Shi is with the Key Laboratory of Intelligent Computing and Signal Processing, Ministry of Education, School of Mathematical Sciences, Anhui University, Hefei, China (email: smjwcl.good@163.com).}, 
Shitao Li\thanks{Shitao Li is with the School of Internet, Anhui University, Hefei, Anhui 230039, China (email: lishitao0216@163.com).},
Yuhong Xia\thanks{Yuhong Xia is with the School of Mathematical Sciences, Anhui University, Hefei, China (email: yhxia88@163.com).}, 
Tor Helleseth\thanks{Tor Helleseth is with Department of Informatics, University of Bergen, Bergen, Norway (email: tor.helleseth@uib.no).},
Ferruh \"{O}zbudak\thanks{Ferruh \"{O}zbudak is with the Faculty of Engineering and Natural Sciences, Sabanc{\i} University, 34956 Istanbul, Turkiye (email: ferruh.ozbudak@sabanciuniv.edu).}}

\date{}
    \maketitle
    
\begin{abstract}
	Let $q=3^m$, let $K=\mathbb F_{q^2}$, and let
	\[
	\mathcal T=\{x\in K^*:\Norm_{K/\mathbb F_q}(x)=1\}.
	\]
For both cyclic and constacyclic Gashkov-Sidel'nikov codes, we show that the set of signed parity-check column labels is precisely $\mathcal T$. Consequently, the decoding problem separates into two stages: determining the minimum error weight associated with a syndrome $S$ and constructing an error vector attaining this minimum. 
We identify the former quantity with the minimum additive length of $S$ with respect to $\mathcal T$ and determine it exactly by the norm and the quadratic character of $\mathbb F_q$. 
We also determine the complete coset-weight distribution and recover the known covering
radius $3$. For the constructive part, we use quadratic-character sums and Weil bounds to construct a coset leader for every syndrome of coset weight three. The resulting procedures give complete maximum-likelihood decoders.
\end{abstract}

\noindent\textbf{Keywords.}
Gashkov-Sidel'nikov code, norm-one torus, additive
decomposition, syndrome decoding, quadratic character.

\medskip
\noindent\textbf{Mathematics Subject Classification (2020).}
11T06, 11T23, 11T71, 94B35.

\section{Introduction}
For a linear code with parity-check matrix $P$, nearest-neighbor decoding can be formulated as the problem of finding a minimum-weight error vector with a prescribed syndrome. Given a received word with syndrome $S$, one seeks an error vector ${\bf e}$ of minimum Hamming weight satisfying
${\bf e}P^{\mathsf T}=S.$
For a syndrome $S$, define the associated coset weight
\begin{equation*}
\omega_P(S)=\min\{\wt(\boldsymbol e):\boldsymbol eP^{\mathsf T}=S\}.
\end{equation*}
The decoding problem naturally separates into two problems. The first is the \emph{value problem} of determining $\omega_P(S)$; the second is the \emph{search problem} of constructing an error vector $\boldsymbol e_S$ satisfying
\begin{equation}\label{eq:intro-search-problem}
	\boldsymbol e_SP^{\mathsf T}=S ~{\rm and}~
	\wt(\boldsymbol e_S)=\omega_P(S).
\end{equation}
Over the ternary symmetric channel with crossover probability $0<p<2/3$, this is equivalent to maximum-likelihood decoding. For general linear codes, finding a minimum-weight vector with a prescribed syndrome is hard \cite{NP-Hard}. However, the two stages above can sometimes be converted into additive problems for code families with substantial algebraic structure. The purpose of this paper is to develop such a reformulation for the ternary Gashkov-Sidel'nikov codes and to solve both the resulting value and search problems.

In 1986, Gashkov and Sidel'nikov \cite{GS} introduced two families of ternary double-error-correcting codes. For $q=3^m$, both families have parameters
\[
\left[\frac{q+1}{2},\,\frac{q+1}{2}-2m,\,5\right]_3
\]
and covering radius $3$. The cyclic family occurs for even $m$, whereas the constacyclic family occurs for odd $m$; we denote the two families by $C_m$ and $D_m$, respectively.
Let $K=\F_{q^2}$, write $\overline{x}=x^q$. Fix an $\F_3$-linear isomorphism
\[
\phi:K\longrightarrow\F_3^{2m}.
\]
For a column-label vector $p=(p_0,\ldots,p_{n-1})\in K^n$, define the ordinary ternary matrix
\[
P_\phi=\bigl(\phi(p_0)\ \phi(p_1)\ \cdots\ \phi(p_{n-1})\bigr)
\in\F_3^{2m\times n}.
\]
For $\boldsymbol e=(e_0,\ldots,e_{n-1})\in\F_3^n$, we identify the ordinary syndrome $\boldsymbol eP_\phi^{\mathsf T}$ with the $K$-valued syndrome
\begin{equation}\label{eq:K-valued-syndrome}
	\sigma_p(\boldsymbol e)=\sum_{j=0}^{n-1}e_jp_j\in K.
\end{equation}
This identification is fixed throughout the paper.

\begin{definition}[Cyclic Gashkov-Sidel'nikov code \cite{GS}]\label{cyclic}
	Assume that $m\geq2$ is even and put $n=(q+1)/2$. Let $\beta\in K^*$ have order $n$, and set
	\[
	p_C=(1,\beta,\beta^2,\ldots,\beta^{n-1}).
	\]
	The cyclic Gashkov--Sidel'nikov code $C_m$ is the ternary code with parity-check matrix
	\[
	P_C=\bigl(\phi(1)\ \phi(\beta)\ \cdots\ \phi(\beta^{n-1})\bigr).
	\]
\end{definition}

\begin{definition}[Constacyclic Gashkov-Sidel'nikov code \cite{GS}]\label{constacyclic}
	Assume that $m\geq3$ is odd and put $n=(q+1)/2$. Let $\beta\in K^*$ have order $n$, and let $\theta\in K^*$ have order $4$. Set
	\[
	p_D=\left(
		1, \beta, \cdots,\beta^{n/2 -1} , \theta , \theta \beta , \cdots , \theta \beta^{n/2-1} 
		\right).
	\]
	The constacyclic Gashkov--Sidel'nikov code $D_m$ is the ternary code with parity-check matrix
	\[
	P_D=\bigl(\phi(1)\ \phi(\beta)\ \cdots\ \phi(\beta^{n/2 -1})\ 
	\phi(\theta)\ \phi(\theta \beta)\ \cdots \ \phi( \theta \beta^{n/2-1} )\bigr).
	\]
\end{definition}

In the nontrivial ranges considered here, namely even $m\geq2$ and odd $m\geq3$, the codes are quasi-perfect. They are closely related to the Zetterberg codes \cite{Z}, whose quasi-perfectness and algebraic decoding were studied in \cite{SM,K,DN}. 
The determination of covering radii and coset leaders for structured code families remains an active topic in coding theory. 
For highly structured algebraic codes, however, exact covering-radius and deep-hole problems continue to be studied using finite-field, character-sum, and geometric methods; see, for example,\cite{BBB2023,DMP2025,FXZ2025,SHO,SLHO,XY}. The covering radius of the original ternary families was already known and is not used as an input below. Our objective is instead to determine the coset weight for every syndrome and to construct an explicit coset leader attaining that weight. Such a complete constructive description has not previously been available for both original ternary families.

The main point is that both families are governed by the same finite-field object. Let
\[\Norm(x)=x\overline{x}\]
be the norm from $K$ to $\F_q$. The norm-one group
\[
\Torus=\{x\in K^*: \Norm(x)=1\}
\]
is the group of $\F_q$-rational points of the one-dimensional norm-one torus and has order $q+1$. 
In the cyclic case, the columns of the parity-check matrix form the index-two subgroup of $\Torus$, and adjoining their negatives gives all of $\Torus$.  In the constacyclic case, the columns and their negatives again partition $\Torus$.  Consequently, in both families there is a bijection between signed coordinate positions and elements of $\Torus$. This observation turns syndrome decoding into a problem about additive length in $\Torus$.

First, we study the value problem of determining $\omega_P(S)$.
For $S\in K$, define the additive length
\[
\ellT(S)=\min\{r\geq0:S=t_1+\cdots+t_r,\ t_i\in\Torus\},
\]
where the empty sum is $0$. For cyclic and constacyclic Gashkov-Sidel'nikov codes, we prove 
$$\omega_P(S)=\ellT(S).$$ 
Thus the value problem is equivalent to determining the additive length of $S$ with respect to $\Torus$. Our main result gives the complete answer to this additive problem.

\begin{theorem}\label{thm:main}
Let $q=3^m$ with $m\geq1$, let $K=\F_{q^2}$, and let $\chi$ be the quadratic character of $\F_q$, extended by $\chi(0)=0$.
	\begin{enumerate}
		\item [(1)] For every $S\in K$,
		\[
		\ellT(S)=
		\begin{cases}
			0, & S=0,\\[1mm]
			1, & \Norm(S)=1,\\[1mm]
			2, & S\neq0,\ \Norm(S)\neq1,\text{ and }
			\chi\!\left(1-\Norm(S)^{-1}\right)=-1,\\[1mm]
			3, & \text{otherwise}.
		\end{cases}
		\]
		\item [(2)] If $\ellT(S)=2$, then the two summands are the distinct roots in $\Torus$ of
		\begin{equation}\label{eq:intro-quadratic}
			f_S(X)=X^2-SX+\frac{S}{\overline S}.
		\end{equation}
		The decomposition is unique up to the ordering of the two roots.
		
		\item [(3)] The sumsets satisfy
		\[
		|\Torus+\Torus|=\frac{q^2+2q+3}{2},
		\quad
		\Torus+\Torus+\Torus=K.
		\]
		The numbers of elements of additive lengths $0,1,2$, and $3$ are, respectively,
		\[
		1,\quad q+1,\quad \frac{q^2-1}{2},\quad
		\frac{(q-3)(q+1)}{2}.
		\]
	\end{enumerate}
\end{theorem}

The proof has two parts. For lengths at most two, conjugation in $K/\F_q$ converts a two-term decomposition into the quadratic equation \eqref{eq:intro-quadratic}; its normalized discriminant is $1-\Norm(S)^{-1}$. This gives the exact length-two criterion and the size of $\Torus+\Torus$. For the remaining case, we study the map
\[
\beta\longmapsto \Norm(S-\beta),
\quad \beta\in\Torus.
\]
Every fiber has size at most two. A cardinality argument then produces a first summand $\beta$ for which $S-\beta$ has length two, and the remaining two summands are obtained from the same quadratic equation.

Theorem~\ref{thm:main} has the following coding-theoretic consequence.

\begin{cor}\label{cor:main-codes}
	Let $q=3^m$. For $m\geq2$ even, let $C_m$ be the cyclic Gashkov-Sidel'nikov code, and for $m\geq3$ odd, let $D_m$ be the constacyclic Gashkov-Sidel'nikov code. For either family, the coset weight associated with a syndrome $S$ is $\ellT(S)$. Consequently, the numbers of cosets of minimum weights $0,1,2,$ and $3$ are
	\[
	1,\quad q+1,\quad \frac{q^2-1}{2},\quad \frac{(q-3)(q+1)}{2},
	\]
	respectively. In particular, both code families have covering radius $3$ and are quasi-perfect.
\end{cor}

Theorem~\ref{thm:main} solves the value problem of determining $\omega_P(S)$. We next solve the search problem of constructing an error vector $\boldsymbol e_S$ satisfying \eqref{eq:intro-search-problem}. For coset weights at most two, a coset leader is recovered from the norm test, the roots of \eqref{eq:intro-quadratic}, and the inverse signed-column map. For coset weight three, the proof of the three-fold sumset identity already gives a direct search-based construction. Algorithms~\ref{alg:cyclic} and \ref{alg:constacyclic} retain this constructive objective but organize the three-term decomposition through one-dimensional searches in conic coordinates. Their parameter choices are controlled by quadratic-character sums and Weil bounds. The resulting algorithms satisfy the following statement.
 
\begin{theorem}\label{thm:main-decoding}
	Let $S$ be a syndrome.
	\begin{enumerate}
		\item [(1)] For every even $m\geq2$, the procedure using Lemma~\ref{lem:short-decomposition} and Remark~\ref{rem:inverse-signed-columns} when $\omega_{P_C}(S)\le2$, and Algorithm~\ref{alg:cyclic} when $\omega_{P_C}(S)=3$, outputs an error vector $\boldsymbol e_S\in\F_3^n$ such that
		\[
		\sigma_{p_C}(\boldsymbol e_S)=S,
		\quad
		\wt(\boldsymbol e_S)=\omega_{P_C}(S).
		\]
		\item [(2)] For every odd $m\geq3$, the procedure using Lemma~\ref{lem:short-decomposition} and Remark~\ref{rem:inverse-signed-columns} when $\omega_{P_D}(S)\le2$, and Algorithm~\ref{alg:constacyclic} when $\omega_{P_D}(S)=3$, outputs an error vector $\boldsymbol e_S\in\F_3^n$ such that
		\[
		\sigma_{p_D}(\boldsymbol e_S)=S,
		\quad
		\wt(\boldsymbol e_S)=\omega_{P_D}(S).
		\]
	\end{enumerate}
	For a received vector $\boldsymbol r$, let $p=p_C$ in part~\textnormal{(1)} and $p=p_D$ in part~\textnormal{(2)}. The decoded codeword is
	\[
	\widehat{\boldsymbol c}
	=\boldsymbol r-\boldsymbol e_{\sigma_p(\boldsymbol r)}.
	\]
	Consequently, both procedures are complete nearest-neighbor decoders. Over the ternary symmetric channel with $0<p<2/3$, they are also maximum-likelihood decoders.
\end{theorem}

Theorem~\ref{thm:main} is the structural core of the paper, whereas Theorem~\ref{thm:main-decoding} makes the search stage explicit. The common additive geometry is independent of the parity of $m$; the code-specific differences occur in the normalization and in the final conversion of torus summands into coordinate positions and nonzero error values.

The paper is organized as follows. Section~\ref{sec2} establishes the signed-column model, and proves Theorem~\ref{thm:main} and Corollary~\ref{cor:main-codes}. Section~\ref{sec3} develops Algorithm~1 and proves Theorem~\ref{thm:main-decoding} (1). Section~\ref{sec4} develops Algorithm~2 and proves Theorem~\ref{thm:main-decoding} (2). Section~\ref{sec5} concludes the paper.

\section{The norm-one torus and proof of the main theorem}\label{sec2}

Let $\F_{q_0}$ denote the {\em Galois field} with $q_0$ elements, and let $\F_{q_0}^*$ denote its multiplicative group. The {\em (Hamming) weight} of a vector  ${\bf c} \in \F_{q_0}^n$ is the number of nonzero components in ${\bf c}$. A {\em $q_0$-ary linear $[n,k,d]$ code} $C$ is a $k$-dimensional subspace of $\F_{q_0}^n$ with {\em minimum distance} $d$, where $d$ is the minimum nonzero weight of all $n$-tuples (called
codewords) in $C$.
For any vector ${\bf a} \in \F_{q_0}^n$, the {\em coset} of $C$ determined by ${\bf a}$ is defined by
$${\bf a} + C = \{{\bf a} + {\bf v} ~|~ {\bf v}\in C\}.$$
A {\em coset leader} of a coset is a vector of smallest weight in the coset, the {\em weight} of a coset is defined as the weight of its coset leader. The {\em packing radius} $t(C)$ of a linear $[n,k,d]$ code $C$ is given by 
$t(C)=\left\lfloor\frac{d-1}{2}\right\rfloor$.
The {\em covering radius} $\rho(C)$ of $C$ is the smallest integer $\rho$ such that the spheres of radius $\rho$ around the codewords cover the Hamming space $\F_q^n$. Equivalently, $\rho(C)$ is also the maximum weight among all coset leaders of the code. The covering radius is a fundamental geometric invariant of a code and is closely connected to objects in finite geometry, such as complete caps, blocking sets, and saturating sets in projective spaces; see, for example, \cite{BLMS,LMS,HS,JLMS-1,JLMS-2}.

The main goal of this paper is to develop a decoding algorithm for ternary Gashkov-Sidelnikov codes. We first review the main principles of maximum-likelihood decoding; further details can be found in the classical monographs \cite{LX,MS}. Let $C$ be a linear code with parity-check matrix $P$. Assume the codeword ${\bf v}\in C$ is transmitted and the word {\bf w} is received, resulting in the error pattern 
$${\bf e} = {\bf w} -{\bf v} \in {\bf w} + C.$$
Then ${\bf w}-{\bf e} = {\bf v} \in C$, and both the error pattern ${\bf e}$ and the received word ${\bf w}$ are in the same coset. Note that error patterns of small weight are most likely under the nearest-neighbor decoding rule for linear codes. Hence, upon receiving the word ${\bf w}$, we choose a word ${\bf e}$ of least weight (that is the coset leader) in the coset ${\bf w} + C$ and conclude that ${\bf v} = {\bf w}-{\bf e}$ was the codeword transmitted.

For any ${\bf u}\in \F_{q}^n$, the {\em syndrome} of ${\bf u}$ is the word $S_{\bf u} = {\bf u} P^T$. 
It follows that the syndrome of a received vector ${\bf w}$ is $S_{\bf w} = {\bf w} P^T={\bf e}P^T$ since ${\bf v}P^T={\bf 0}$. By \cite[Remark 4.8.12]{LX}, there is a one-to-one correspondence between cosets and syndromes. A complete decoding algorithm identifies the coset leader corresponding to the received syndrome, which represents the most likely error.

\subsection{Short decompositions and sumsets}

Throughout this paper, let $q=3^m$, let $K=\F_{q^2}$, and write $\overline{x}=x^q$. The norm map is $\Norm(x)=x\overline{x}$, and
\[
\Torus=\{x\in K^*:x^{q+1}=1\}
\]
is the cyclic norm-one group of order $q+1$. We write
\[
\TorusZero=\{t^2:t\in\Torus\}
\]
for its index-two subgroup.

\begin{prop}\label{prop:signed-columns}
	For either $C_m$ or $D_m$, let $p_0,\ldots,p_{n-1}$ be the corresponding column labels. The map
	\begin{equation}\label{eq:signed-column-map}
		\iota_P:\{0,\ldots,n-1\}\times\F_3^*
		\longrightarrow\Torus,
		\quad
		(j,u)\longmapsto up_j,
	\end{equation}
	is a bijection.
	
	More precisely, in the cyclic case the column-label set is $\TorusZero$ and
	\[
	\Torus=\TorusZero\sqcup(-\TorusZero).
	\]
	In the constacyclic case, with
	\[
	\mathcal B=\{1,\beta,\ldots,\beta^{n/2-1}\},
	\quad
	\mathcal P_D=\mathcal B\sqcup\theta\mathcal B,
	\]
	one has
	\[
	\TorusZero=\mathcal B\sqcup(-\mathcal B)
	\quad\text{and}\quad
	\Torus=\mathcal P_D\sqcup(-\mathcal P_D)
	=\mathcal B\sqcup\theta\mathcal B\sqcup(-\mathcal B)\sqcup(-\theta\mathcal B).
	\]
\end{prop}

\begin{proof}
	For $C_m$, the column-label set is $\langle\beta\rangle=\TorusZero$ and has order $(q+1)/2$. Since $m$ is even, $(q+1)/2$ is odd, so $-1\notin\TorusZero$. Hence $\Torus=\TorusZero\sqcup(-\TorusZero)$.
	
	For $D_m$, one has $\beta^{n/2}=-1$, and therefore
	\[
	\TorusZero=\langle\beta\rangle
	=\mathcal B\sqcup(-\mathcal B).
	\]
	Moreover, $n\equiv2\pmod4$, so $\theta\notin\TorusZero$. Hence
	\[
	\Torus=\TorusZero\sqcup\theta\TorusZero
	=\mathcal B\sqcup(-\mathcal B)\sqcup\theta\mathcal B\sqcup(-\theta\mathcal B).
	\]
	The first and third parts of the last disjoint union are precisely the labels of the columns of $P_D$. In either family there are $2n=q+1$ signed columns, so the displayed map is a bijection.
\end{proof}

For either family, let $P$ denote the relevant ternary parity-check matrix and let $p$ denote its column-label vector. In the $K$-valued model, the coset weight associated with $S$ is
\begin{equation}\label{eq:syndrome-weight}
	\omega_P(S)
	=\min\{\wt(\boldsymbol e):\boldsymbol e\in\F_3^n,
	\ \sigma_p(\boldsymbol e)=S\}.
\end{equation}
Because the kernel of $\sigma_p$ is the code, $\omega_P(S)$ is exactly the weight of the coset corresponding to $S$.

\begin{definition}\label{def:additive-length}
	For $S\in K$, define
	\[
	\ellT(S)=\min\{r\geq0:S=t_1+\cdots+t_r,\ t_i\in\Torus\},
	\]
	where the empty sum is $0$.
\end{definition}

\begin{prop}\label{prop:coset-length}
	For either Gashkov--Sidel'nikov family and every syndrome $S\in K$,
	\[
	\omega_P(S)=\ellT(S).
	\]
\end{prop}

\begin{proof}
	If $\boldsymbol e$ has weight $r$, then \eqref{eq:K-valued-syndrome} expresses its syndrome as a sum of $r$ signed column labels. Proposition~\ref{prop:signed-columns} therefore gives a representation of $S$ by $r$ elements of $\Torus$, so $\ellT(S)\leq\omega_P(S)$.
	
	Conversely, consider a representation of $S$ of minimum length. If two summands are opposite, they cancel. If two summands are equal, then in characteristic $3$ their sum is $2t=-t\in\Torus$, so the two terms can be replaced by one term. Hence a minimum representation contains neither equal nor opposite summands. Under the bijection \eqref{eq:signed-column-map}, its summands correspond to distinct coordinate positions and define an error vector of the same weight. Thus $\omega_P(S)\leq\ellT(S)$.
\end{proof}

\begin{remark}\label{rem:inverse-signed-columns}
	For $t\in\Torus$, write
	\[
	\iota_P^{-1}(t)=\bigl(j_P(t),u_P(t)\bigr),
	\]
	so that
	\[
	t=u_P(t)p_{j_P(t)},
	\quad
	j_P(t)\in\{0,\ldots,n-1\},
	\quad
	u_P(t)\in\F_3^*.
	\]
	For the cyclic family,
	\[
	u_{P_C}(t)=
	\begin{cases}
		1, & t\in\TorusZero,\\
		-1, & t\in-\TorusZero,
	\end{cases}
	\qquad
	\beta_C^*(t)=u_{P_C}(t)t\in\TorusZero,
	\]
	and $j_{P_C}(t)$ is the unique index satisfying
	\[
	\beta_C^*(t)=\beta^{j_{P_C}(t)}.
	\]
	For the constacyclic family,
	\[
	u_{P_D}(t)=
	\begin{cases}
		1, & t\in\mathcal P_D,\\
		-1, & t\in-\mathcal P_D,
	\end{cases}
	\qquad
	\beta_D^*(t)=u_{P_D}(t)t\in\mathcal P_D,
	\]
	and $j_{P_D}(t)$ is the unique index satisfying
	\[
	\beta_D^*(t)=
	\begin{cases}
		\beta^{j_{P_D}(t)}, & 0\le j_{P_D}(t)<n/2,\\
		\theta\beta^{j_{P_D}(t)-n/2}, & n/2\le j_{P_D}(t)<n.
	\end{cases}
	\]
	Thus $j_P(t)$ is the recovered coordinate position and $u_P(t)$ is the corresponding nonzero ternary error value. These inverse maps may be stored as precomputed column-label dictionaries. They will be used explicitly in Step~5 of each branch of Algorithms~\ref{alg:cyclic} and \ref{alg:constacyclic}.
\end{remark}

Let $\chi$ denote the quadratic character of $\F_q$, extended by $\chi(0)=0$. Specifically, the quadratic character $\chi$ on $\F_q$ is given by
\be
\begin{array}{rcl}
	\chi:\F_q & \ra & \{0,1,-1\} \\
	x & \mapsto & \left\{
	\begin{array}{ll}
		0, & \mbox{if $x=0$}, \\
		1, & \mbox{if $x \in \F_q^*$ is a square}, \\
		-1, & \mbox{if $x \in \F_q^*$ is a nonsquare}.
	\end{array}
	\right.
\end{array}
\nn\ee
The following lemma gives one- and two-term decompositions.

\begin{lem}\label{lem:short-decomposition}
	For any $S\in K$, the following statements hold.
	\begin{enumerate}
		\item[(1)] $\ellT(S)=0$ if and only if $S=0$.

		\item[(2)] $\ellT(S)=1$ if and only if $\Norm(S)=1$.
		
		\item[(3)] Suppose that $S\neq0$ and $\Norm(S)\neq1$. Then
		\[
		\ellT(S)=2
		\quad\Longleftrightarrow\quad
		\chi\!\left(1-\Norm(S)^{-1}\right)=-1.
		\]
		In this case, the two summands are precisely the two distinct roots in
		$\Torus$ of
		\[
		f_S(X)=X^2-SX+\frac{S}{\overline S}.
		\]
		More explicitly, if
		\[
		\delta^2=1-\Norm(S)^{-1},
		\]
		then these roots are
		\[
		t_+=-S+S\delta
		\quad\text{and}\quad
		t_-=-S-S\delta.
		\]
	\end{enumerate}
\end{lem}

\begin{proof}
	The assertion in part~(1) is immediate. For $S\neq0$, one has
	$\ellT(S)=1$ precisely when $S\in\Torus$, which is equivalent to
	$\Norm(S)=1$. This proves part~(2).
	
	We now prove part~(3). Suppose first that $\ellT(S)=2$. Then there exist
	$t_1,t_2\in\Torus$ such that
	\[
	S=t_1+t_2.
	\]
	Since $\Norm(S)\neq1$, the two summands must be distinct. Indeed, if
	$t_1=t_2$, then
	\[
	S=2t_1=-t_1\in\Torus,
	\]
	contrary to $\Norm(S)\neq1$.
	Since $\overline{t_i}=t_i^{-1}$, we have
	\[
	\overline S
	=t_1^{-1}+t_2^{-1}
	=\frac{t_1+t_2}{t_1t_2}
	=\frac{S}{t_1t_2}.
	\]
	Consequently,
	\[
	t_1t_2=\frac{S}{\overline S},
	\]
	and hence $t_1$ and $t_2$ are the two distinct roots of
	\[
	f_S(X)=X^2-SX+\frac{S}{\overline S}.
	\]
	
	Set $\delta=\frac{t_1-t_2}{S}.$
	Using $\overline{t_i}=t_i^{-1}$ and $t_1t_2\overline S=S$, we obtain
	\[
	\overline\delta
	=\frac{t_1^{-1}-t_2^{-1}}{\overline S}
	=-\frac{t_1-t_2}{t_1t_2\overline S}
	=-\frac{t_1-t_2}{S}
	=-\delta.
	\]
	Since $t_1\neq t_2$, we have $\delta\neq0$. As the characteristic is odd,
	the relation $\overline\delta=-\delta$ implies that
	$\delta\notin\F_q$.
	Moreover, the discriminant identity	gives
\[
\delta^2
=\frac{(t_1-t_2)^2}{S^2}
=1-\frac{t_1t_2}{S^2}
=1-\Norm(S)^{-1}.
\]
Thus $\delta^2\in\F_q^*$ but $\delta\notin\F_q$. Therefore
	$\delta^2$ is a nonsquare in $\F_q^*$, and hence
	\[
	\chi\left(1-\Norm(S)^{-1}\right)=-1.
	\]
	
	Conversely, suppose that $S\neq0$, $\Norm(S)\neq1$, and
$\chi\left(1-\Norm(S)^{-1}\right)=-1.$
	Put 
	\[d=1-\Norm(S)^{-1}\in\F_q^*.\]
	Since $d$ is a nonsquare in $\F_q^*$, there exists
	$\delta\in K\setminus\F_q$ such that
	$\delta^2=d.$
	Because $\delta^2\in\F_q$, we have
	$\overline\delta^2=\delta^2.$
	As $\delta\notin\F_q$, it follows that
	$\overline\delta=-\delta.$ Define
	\[
	t_+=-S+S\delta
	\quad\text{and}\quad
	t_-=-S-S\delta.
	\]
	Since the characteristic is $3$, we have
	$t_++t_-=-2S=S.$
	Furthermore,
\[
t_+t_-=S^2(1-\delta^2)
=\frac{S^2}{\Norm(S)}
=\frac{S}{\overline S}.
\]
Thus $t_+$ and $t_-$ are the two roots of $f_S(X)$. They are distinct
	because $S\neq0$ and $\delta\neq0$.
	It remains to verify that these roots belong to $\Torus$. Using
	$\overline\delta=-\delta$, we obtain
\begin{align*}
	\Norm(t_\pm) &= \Norm(S) (-1 \pm \delta)(-1\pm \overline{\delta})   \\
	&= \Norm(S)(1 \mp (\delta + \overline{\delta}) + \overline{\delta} \delta)\\
	&= \Norm(S)(1 + \overline{\delta} \delta)  \\
	&= \Norm(S)(1 - \delta^2 )                       \\
	&=  1,
\end{align*}
and hence $t_+,t_-\in\Torus$. Hence $f_S$ has two distinct roots in $\Torus$.
	Therefore $\ellT(S)\leq2$. Since $S\neq0$ and $\Norm(S)\neq1$, parts~(1)
	and~(2) exclude lengths zero and one. Consequently,
	\[
	\ellT(S)=2.
	\]
\end{proof}

\begin{lem}\label{lem:sumsets}
	One has
	\[
	|\Torus+\Torus|=\frac{q^2+2q+3}{2}
	\quad\text{and}\quad
	\Torus+\Torus+\Torus=K.
	\]
\end{lem}

\begin{proof}
By Lemma~\ref{lem:short-decomposition}, an element $S$ has exact additive length two precisely when its nonzero norm $a=\Norm(S)$ belongs to
\[
A=\{a\in\F_q^*\setminus\{1\}:\chi(1-a^{-1})=-1\}.
\]
The map $a\mapsto1-a^{-1}$ is a bijection from $A$ to the set of nonsquares in $\F_q^*$, so $|A|=(q-1)/2$. Every nonzero norm fiber in $K$ has size $q+1$. Hence the number of elements of exact length two is
\[
\frac{q-1}{2}(q+1)=\frac{q^2-1}{2}.
\]
Together with $0$ and the $q+1$ elements of $\Torus$, this gives
\[
|\Torus+\Torus|
=1+(q+1)+\frac{q^2-1}{2}
=\frac{q^2+2q+3}{2}.
\]

Now let $S\notin\Torus+\Torus$. Since $t=(-t)+(-t)$ for every $t\in\Torus$ in characteristic $3$, one has $\Torus\subseteq\Torus+\Torus$, and therefore $S\notin\Torus$. Define
\[
\varphi_S:\Torus\longrightarrow\F_q^*,
\quad
\varphi_S(\beta)=\Norm(S-\beta).
\]
The codomain is $\F_q^*$ because $S\notin\Torus$. For $a\in\F_q^*$, the equation $\varphi_S(\beta)=a$ is equivalent to
\[
\overline S\,\beta^2-(\Norm(S)+1-a)\beta+S=0.
\]
Thus every fiber has size at most two, and
\[
|{\rm Im}(\varphi_S)|\geq\frac{q+1}{2}.
\]
Let
\[
G=\{a\in\F_q^*: \chi(1-a^{-1})=-1\}.
\]
Then $|G|=(q-1)/2$. Both $G$ and ${\rm Im}(\varphi_S)$ lie in the $(q-1)$-element set $\F_q^*$, while their cardinalities sum to at least $q$. Hence they intersect. Choose $\beta_1\in\Torus$ such that $\Norm(S-\beta_1)\in G$, and put $B=S-\beta_1$. By the proof of Lemma~\ref{lem:short-decomposition}, the polynomial
\[
X^2-BX+\frac{B}{\overline B}
\]
has two distinct roots $\beta_2,\beta_3\in\Torus$. Since $\beta_2+\beta_3=B$, we obtain
\[
S=\beta_1+\beta_2+\beta_3.
\]
Thus every element of $K$ is a sum of three elements of $\Torus$.
\end{proof}

Next, we present a proof of Theorem~\ref{thm:main}.

\begin{proof}[Proof of Theorem~\ref{thm:main}]
	Lemma~\ref{lem:short-decomposition} proves the assertions for additive lengths zero, one, and two, including the quadratic equation for the two summands. Lemma~\ref{lem:sumsets} gives both
	\[
	|\Torus+\Torus|=\frac{q^2+2q+3}{2}
	\qquad\text{and}\qquad
	\Torus+\Torus+\Torus=K.
	\]
	Hence every element not covered by the first three cases has additive length exactly three.
	
	There is one element of length zero and $q+1$ elements of length one. In the proof of Lemma~\ref{lem:sumsets}, the elements of exact length two were counted as
	\[
	\frac{q-1}{2}(q+1)=\frac{q^2-1}{2}.
	\]
	The number of remaining elements is therefore
	\[
	q^2-1-(q+1)-\frac{q^2-1}{2}
	=\frac{(q-3)(q+1)}{2},
	\]
	which proves the final assertion.
\end{proof}

Next, we present a proof of Corollary~\ref{cor:main-codes}.

\begin{proof}[Proof of Corollary~\ref{cor:main-codes}]
	Both $C_m$ and $D_m$ have codimension $2m$, and hence each has $q^2$ syndromes. Proposition~\ref{prop:signed-columns} identifies the signed parity-check columns of either family with $\Torus$, while Proposition~\ref{prop:coset-length} identifies the minimum coset weight with $\ellT(S)$. The coset-weight distribution is therefore exactly the additive-length distribution in Theorem~\ref{thm:main}.
	
	In the ranges $m\geq2$ even and $m\geq3$ odd, one has $q>3$, so
	\[
	\frac{(q-3)(q+1)}{2}>0.
	\]
	Thus weight-three cosets exist, while every coset has weight at most three. The covering radius is therefore exactly three. Since both code families have minimum distance five, their packing radius is two, and they are quasi-perfect. This gives an alternative derivation of the known covering radius from the additive structure of $\Torus$.
\end{proof}

\subsection{Coset-leader recovery}
The preceding results determine the coset weight associated with every syndrome. We now turn to the search problem and convert a minimum-length torus decomposition into an error vector.

\begin{example}
	Assume that $m=2$ and $n=(3^2+1)/2=5$. 
	Let
	\[
	\F_{3^4}=\F_3[\omega]/(\omega^4+2\omega^3+2),
	\]
	where $\omega$ is primitive.
Let $\beta=\omega^{16}$ have order $5$, and let $C_2$ have column labels
	\[
	(1,\beta,\beta^2,\beta^3,\beta^4).
	\]
	Let $\boldsymbol c=(1,1,1,1,1)$ be a codeword of $C_2$.
	\begin{itemize}
		\item For $\boldsymbol r=(1,0,1,1,1)$, the syndrome is $S=\omega^{56}=-\beta$. Since $\Norm(S)=1$, the associated coset has weight one. Remark~\ref{rem:inverse-signed-columns} gives coordinate $1$ (the second coordinate) and error value $-1$ under zero-based indexing, so
		\[
		\boldsymbol e=(0,-1,0,0,0),
		\quad
		\boldsymbol r-\boldsymbol e=\boldsymbol c.
		\]
		\item For $\boldsymbol r=(1,0,2,1,1)$, the syndrome is $S=\omega^{39}$. One computes
		\[
		1-\Norm(S)^{-1}=\omega^{30},
		\]
		which is a nonsquare in $\F_{3^2}^*$. Choose $\delta=\omega^{15}$, so that $\delta^2=1-\Norm(S)^{-1}$. The two roots of $f_S$ are
		\[
		t_+=-S+S\delta=\omega^{56}=-\beta,
		\quad
		t_-=-S-S\delta=\omega^{32}=\beta^2.
		\]
		They correspond to coordinates $1$ and $2$ (the second and third coordinates) with error values $-1$ and $1$, respectively. Thus
		\[
		\boldsymbol e=(0,-1,1,0,0),
		\quad
		\boldsymbol r-\boldsymbol e=\boldsymbol c.
		\]
	\end{itemize}
\end{example}

Finally, we give a direct algebraic decomposition procedure.

\begin{prop}\label{prop:direct-decoder}
Let $P$ be the parity-check matrix of either Gashkov--Sidel'nikov family, let $p$ be its column-label vector, and let $S\in K$. A coset leader $\boldsymbol e_S$ can be obtained by the following sequential procedure.
	\begin{enumerate}
		\item [(1)] If $S=0$, take the empty decomposition.
		\item [(2)] If $S\neq0$ and $\Norm(S)=1$, take $t_1=S$.
		\item [(3)] If $S\neq0$, $\Norm(S)\neq1$, and $\chi(1-\Norm(S)^{-1})=-1$, take the two roots $t_1,t_2$ of $f_S$.
		\item [(4)] Otherwise, search over $\beta\in\Torus$ until $B=S-\beta\neq0$ and
		\[
		\chi\left(1-\Norm(B)^{-1}\right)=-1.
		\]
		Such a $\beta$ exists by Lemma~\ref{lem:sumsets}. Take $t_1=\beta$ and let $t_2,t_3$ be the two roots of $f_B$.
	\end{enumerate}
	For each summand, compute
	\[
	\bigl(j_i,u_i\bigr)=\iota_P^{-1}(t_i)
	\]
	as in Remark~\ref{rem:inverse-signed-columns}, and define $\boldsymbol e_S$ by $(\boldsymbol e_S)_{j_i}=u_i$ and zero entries elsewhere. Then
	\[
	\sigma_p(\boldsymbol e_S)=S,
	\quad
	\wt(\boldsymbol e_S)=\omega_P(S).
	\]
\end{prop}

\begin{proof}
	Theorem~\ref{thm:main} and Lemma~\ref{lem:sumsets} guarantee that the procedure terminates and returns a minimum-length decomposition. By Proposition~\ref{prop:coset-length}, no two summands in that decomposition are equal or opposite. Hence their images under the inverse signed-column map have distinct coordinate indices. The definition of $\iota_P$ gives
	\[
	\sigma_p(\boldsymbol e_S)
	=\sum_i u_i p_{j_i}
	=\sum_i t_i=S,
	\]
	and the weight of $\boldsymbol e_S$ is the number of summands, namely $\omega_P(S)$.
\end{proof}

\begin{remark}
After the syndrome has been computed, the length-three part of Proposition~\ref{prop:direct-decoder} tests at most $q+1$ elements of $\Torus$. Algorithms~\ref{alg:cyclic} and \ref{alg:constacyclic} instead use structured searches over $\F_q$ in conic coordinates. No asymptotic improvement over an $O(q)$ search is claimed; the purpose of the latter algorithms is to give explicit finite-field realizations of the three-term decomposition and of the recovery of coordinate positions and error values.
\end{remark}

For the parameter-selection arguments, we use the standard quadratic-character estimate: if $f\in\F_q[X]$ is squarefree of degree $d\geq1$ and is not a constant multiple of a square over an algebraic closure, then
\[
\left|\sum_{x\in\F_q}\chi(f(x))\right|
\leq(d-1)\sqrt q;
\]
see, for example, \cite[Theorem~5.41]{finite}. We also use the fact that a separable quadratic polynomial has a complete quadratic-character sum of absolute value one.

\section{Cyclic three-term decompositions} \label{sec3}
Assume throughout this section that $m\geq2$ is even and put $q=3^m$. The cyclic column-label set is $\TorusZero$, and
\[
\Torus=\TorusZero\sqcup(-\TorusZero).
\]
The map $x\mapsto x^{q-1}$ from $K^*$ onto $\Torus$ is surjective, so we may choose once and for all an element $w_2\in K^*$ satisfying $w_2^{q-1}=-1$. Put
\[
D_2=w_2^{q+1}.
\]
Throughout Sections~\ref{sec3} and \ref{sec4}, division by $2$ is taken in $\F_q$; since the characteristic is $3$, one has $2^{-1}=2$.
Then $w_2^q=-w_2$, $D_2=-w_2^2\in\F_q^*$, and
\[
D_2^{(q-1)/2}
=\bigl(w_2^{q-1}\bigr)^{(q+1)/2}=-1,
\]
because $(q+1)/2$ is odd. Thus $D_2$ is a nonsquare in $\F_q^*$, and
\begin{equation}\label{eq:cyclic-norm-conic}
\Norm(x+w_2y)=(x+w_2y)(x-w_2y)=x^2+D_2y^2
\end{equation}
	for all $x,y\in\F_q$.
	
\begin{algorithm}\label{alg:cyclic}
	\textnormal{Input:} a syndrome $S\in K^*$ with $\ellT(S)=3$.
	
	\textnormal{Output:} an error vector $\boldsymbol e_S\in\F_3^n$ satisfying
	\[
	\sigma_{p_C}(\boldsymbol e_S)=S,
	\quad
	\wt(\boldsymbol e_S)=3.
	\]
	\begin{enumerate}
		\item Set $z=S^{q-1}\in\Torus$.
		\item If $z^{-1}\in\TorusZero$, apply the first cyclic normalization in Subsection~\ref{subsec:cyclic-a}. Otherwise $-z^{-1}\in\TorusZero$, and apply the second cyclic normalization in Subsection~\ref{subsec:cyclic-b}.
		\item In either branch, the construction produces $\beta_1,\beta_2,\beta_3\in\Torus$ with $S=\beta_1+\beta_2+\beta_3$. Compute
		\[
		(j_i,u_i)=\iota_{P_C}^{-1}(\beta_i),
		\quad 1\leq i\leq3,
		\]
		and return the vector $\boldsymbol e_S$ whose entries at $j_1,j_2,j_3$ are $u_1,u_2,u_3$, respectively, and whose remaining entries are zero.
	\end{enumerate}
\end{algorithm}

\subsection{The first cyclic normalization}\label{subsec:cyclic-a}

Retain $z=S^{q-1}$ from Algorithm~\ref{alg:cyclic}. Suppose that $z^{-1}\in\TorusZero$. Put $M=|\TorusZero|=(q+1)/2$. Since $M$ is odd, the squaring map is a permutation of $\TorusZero$, and the unique element $h\in\TorusZero$ satisfying
\begin{equation}\label{eq:cyclic-a-normalization}
	h^2=z^{-1}
\end{equation}
is
\[
h=(z^{-1})^{(M+1)/2}.
\]
Define
\begin{equation}\label{eq:cyclic-a-alpha}
	\alpha=\frac{S}{h}.
\end{equation}
Since $h^{q-1}=h^{-2}=z$, one has $\alpha^{q-1}=1$, and hence $\alpha\in\F_q^*$. Moreover, $\alpha\notin\{1,-1\}$ because $\ellT(S)=3$.
For this $\alpha$, define
\begin{align}
	A(x)&=2\alpha x-\alpha^2-1,\nonumber\\
	B(x)&=2\alpha x^2-x+\alpha^3+\alpha,\nonumber\\
	C(x)&=(-\alpha^2-1)x^2+(\alpha^3+\alpha)x-\alpha^4+\alpha^2,\nonumber\\
	\Delta(x)&=B(x)^2-A(x)C(x).
	\label{eq:cyclic-a-polynomials}
\end{align}
Since the characteristic is $3$, $\Delta(x)=B(x)^2-A(x)C(x)$ is the usual discriminant of the quadratic equation used below.

The first cyclic construction is as follows.
\begin{enumerate}
	\item [{\bf Step 1:}] Search over $x_1\in\F_q^*$ in any fixed order until
	\[
	1-x_1^2\text{ is a nonsquare},\quad
	A(x_1)\neq0,\quad
	\Delta(x_1)\text{ is a nonzero square}.
	\]
	Lemma~\ref{lem:cyclic-a-parameter} proves that the search terminates. Choose $y_1\in\F_q^*$ satisfying
	\begin{equation}\label{eq:cyclic-a-y1}
		D_2y_1^2=1-x_1^2.
	\end{equation}
	\item [{\bf Step 2:}] Solve
	\begin{equation}\label{eq:cyclic-a-x2}
		A(x_1)x_2^2+B(x_1)x_2+C(x_1)=0
	\end{equation}
	in $\F_q$ and choose either of its two roots.
	\item [{\bf Step 3:}] Define
	\begin{equation}\label{eq:cyclic-a-y2}
		y_2=\frac{\alpha(x_1+x_2)-x_1x_2-(\alpha^2+1)/2}{D_2y_1}.
	\end{equation}
	\item [{\bf Step 4:}] Put
	\begin{align}
		h_1&=x_1+w_2y_1,\nonumber\\
		h_2&=x_2+w_2y_2,\nonumber\\
		h_3&=(\alpha-x_1-x_2)+w_2(-y_1-y_2).
		\label{eq:cyclic-a-hi}
	\end{align}
Lemma~\ref{lem:cyclic-a-conic} shows that $h_1,h_2,h_3$ are in $\Torus$.
     \item [{\bf Step 5:}] Set
     \begin{equation}\label{eq:cyclic-a-beta}
     	\beta_i=h_i h,
     	\qquad 1\le i\le3.
     \end{equation}
     By Lemma~\ref{lem:cyclic-a-conic}, the elements $\beta_i$ lie in $\Torus$. Apply the inverse signed-column map of Remark~\ref{rem:inverse-signed-columns}. Explicitly, for $1\le i\le3$, define
     \[
     u_i=
     \begin{cases}
     	1, & \beta_i^{(q+1)/2}=1,\\
     	-1, & \beta_i^{(q+1)/2}=-1,
     \end{cases}
     \qquad
     \beta_i^*=u_i\beta_i\in\TorusZero.
     \]
     Let $j_i$ be the unique coordinate index satisfying
     \[
     \beta_i^*=\beta^{j_i},
     \quad
     0\le j_i<n.
     \]
     Thus $\beta_i=u_i\beta^{j_i}$, so $u_i$ is the recovered error value and $j_i$ is the recovered coordinate position. Define $\boldsymbol e_S\in\F_3^n$ by
     \[
     (\boldsymbol e_S)_{j_i}=u_i
     \quad(1\le i\le3),
     \qquad
     (\boldsymbol e_S)_j=0
     \quad\text{for all other }j.
     \]
     Lemma~\ref{lem:cyclic-a-output} proves that the indices $j_1,j_2,j_3$ are pairwise distinct and that $\boldsymbol e_S$ has syndrome $S$ and weight $3$.
\end{enumerate}

\begin{lem}\label{lem:cyclic-a-parameter}
	Let $q=3^m$ with even $m\geq2$, and let $\alpha\in\F_q^*\setminus\{1,-1\}$. For the polynomials in \eqref{eq:cyclic-a-polynomials}, there exists $x_1\in\F_q^*$ such that $1-x_1^2$ is a nonsquare, $A(x_1)\neq0$, and $\Delta(x_1)$ is a nonzero square.
\end{lem}

\begin{proof}
	Let $\chi$ be the quadratic character of $\F_q$, extended by $\chi(0)=0$, and put
	\[
	g(x)=1-x^2,
	\quad
	Q_\alpha(x)=\alpha x^2+(1-\alpha^2)x+\alpha^3+\alpha.
	\]
	A direct expansion in characteristic $3$ yields
	\begin{equation}\label{eq:cyclic-a-factorization}
		\Delta(x)=B(x)^2-A(x)C(x)=-\alpha g(x)Q_\alpha(x).
	\end{equation}
	The polynomial $g(x)$ is separable with roots $1$ and $-1$, while $Q_\alpha(x)$ has the discriminant
	\[
	\operatorname{disc}(Q_\alpha)=1.
	\]
	Moreover,
	\[
	Q_\alpha(1)=(\alpha-1)^2(\alpha+1),
	\quad
	Q_\alpha(-1)=(\alpha-1)(\alpha+1)^2.
	\]
	Since $\alpha\notin\{0,1,-1\}$, the two quadratic factors in \eqref{eq:cyclic-a-factorization} are separable and coprime. Thus $\Delta$ is a squarefree quartic, and $A$ is a nonzero linear polynomial.
	
	Let
	\[
	E_\alpha=\{0\}\cup\{x\in\F_q:A(x)\Delta(x)=0\}.
	\]
	Then $|E_\alpha|\leq6$. Consider
	\[
	N_1=\sum_{x\in\F_q\setminus E_\alpha}
	\bigl(1-\chi(g(x))\bigr)
	\bigl(1+\chi(\Delta(x))\bigr).
	\]
	For $x\notin E_\alpha$, the summand is $4$ exactly when $g(x)$ is a nonsquare and $\Delta(x)$ is a square, and it is zero otherwise. Hence $N_1/4$ is the number of admissible choices of $x_1$.
	
	Let $N$ denote the same sum over all of $\F_q$. Expanding gives
	\[
	N=q-\sum_{x\in\F_q}\chi(g(x))
	+\sum_{x\in\F_q}\chi(\Delta(x))
	-\sum_{x\in\F_q}\chi(g(x)\Delta(x)).
	\]
	Since $g$ is a separable quadratic,
	\[
	\left|\sum_{x\in\F_q}\chi(g(x))\right|=1.
	\]
	Since $\Delta$ is a squarefree quartic, the Weil bound gives
	\[
	\left|\sum_{x\in\F_q}\chi(\Delta(x))\right|
	\leq3\sqrt q.
	\]
	Finally,
	\[
	g(x)\Delta(x)=-\alpha g(x)^2Q_\alpha(x).
	\]
	Away from the two roots of $g$, its quadratic character is $\chi(-\alpha)\chi(Q_\alpha(x))$. The complete character sum of the separable quadratic $Q_\alpha$ has absolute value one, and deleting at most two values changes the sum by at most two. Therefore
	\[
	\left|\sum_{x\in\F_q}\chi(g(x)\Delta(x))\right|\leq3.
	\]
	It follows that $N\geq q-3\sqrt q-4$. Removing the at most six exceptional points decreases the weighted sum by at most $24$, so
	\[
	N_1\geq q-3\sqrt q-28.
	\]
	This is positive for $q\geq81$. Since $m$ is even, the only remaining field is $q=9$. The Magma computation \cite{Magma} shows that every admissible $\alpha$ has at least two valid choices of $x_1$.
\end{proof}

\begin{lem}\label{lem:cyclic-a-conic}
	The elements in \eqref{eq:cyclic-a-hi} satisfy
	\[
	h_1,h_2,h_3\in\Torus
	\quad\text{and}\quad
	h_1+h_2+h_3=\alpha.
	\]
\end{lem}

\begin{proof}
	Equation \eqref{eq:cyclic-norm-conic} gives
	\[
	\Norm(x+w_2y)=x^2+D_2y^2.
	\]
	Put
	\[
	x_3=\alpha-x_1-x_2,
	\quad
	y_3=-y_1-y_2.
	\]
	Then $h_i=x_i+w_2y_i$ and $h_1+h_2+h_3=\alpha$. It remains to prove
	\[
	x_i^2+D_2y_i^2=1,
	\qquad 1\leq i\leq3.
	\]
	The first equality is \eqref{eq:cyclic-a-y1}. Let
	\[
	R=\alpha(x_1+x_2)-x_1x_2-\frac{\alpha^2+1}{2}.
	\]
	A direct expansion in characteristic $3$ gives
	\begin{equation}\label{eq:cyclic-a-conic-identity}
		A(x_1)x_2^2+B(x_1)x_2+C(x_1)
		=(1-x_1^2)(1-x_2^2)-R^2.
	\end{equation}
	The left-hand side is zero by \eqref{eq:cyclic-a-x2}; hence
	\[
	R^2=(1-x_1^2)(1-x_2^2).
	\]
	Since $D_2y_1^2=1-x_1^2\neq0$ and \eqref{eq:cyclic-a-y2} gives $y_2=R/(D_2y_1)$, we obtain
	\[
	D_2y_2^2
	=\frac{R^2}{D_2y_1^2}
	=1-x_2^2.
	\]
	Thus $\Norm(h_2)=1$. Using the first two norm equations and the definitions of $x_3,y_3$, another direct expansion yields
	\[
	x_3^2+D_2y_3^2-1=R-D_2y_1y_2,
	\]
	which is zero by \eqref{eq:cyclic-a-y2}. Hence $\Norm(h_3)=1$ as well.
\end{proof}

For the case where $z^{-1}=(S^{q-1})^{-1}\in\TorusZero$, we next show that the equations
\[S=\beta_1+\beta_2+\beta_3, ~{\rm and} ~S=u_1\beta_1^* + u_2 \beta_2^* + u_3 \beta_3^*=u_1\beta^{j_1} + u_2 \beta^{j_2} + u_3 \beta^{j_3}
\]
are decompositions of $S$ with respect to $\Torus$ and $\TorusZero$, respectively.

\begin{lem}\label{lem:cyclic-a-output}
	Under the notation of the first cyclic normalization, the elements recovered in Step~5 satisfy
	\[
	\beta_i^*=\beta^{j_i}\in\TorusZero,
	\qquad
	\beta_i=u_i\beta_i^*,
	\qquad 1\le i\le3.
	\]
	The coordinate indices $j_1,j_2,j_3$ are pairwise distinct, and the returned vector $\boldsymbol e_S$ satisfies
	\[
	\sigma_{p_C}(\boldsymbol e_S)=S,
	\quad
	\wt(\boldsymbol e_S)=3.
	\]
\end{lem}

\begin{proof}
	By Lemma~\ref{lem:cyclic-a-conic}, $h_i\in\Torus$ and $h_1+h_2+h_3=\alpha$. Since $h\in\TorusZero\subseteq\Torus$ and $S=\alpha h$,
	\[
	S=(h_1+h_2+h_3)h=\beta_1+\beta_2+\beta_3.
	\]
	Moreover, each $\beta_i=h_i h$ lies in $\Torus$. The disjoint union
	\[
	\Torus=\TorusZero\sqcup(-\TorusZero)
	\]
	implies that exactly one of $\beta_i$ and $-\beta_i$ belongs to $\TorusZero$. 
	If $\beta_i\in \TorusZero$, then Step 5 sets $\beta_i^*=\beta_i$ and $u_i=1$. If $\beta_i\in-\mathcal \TorusZero$, then Step 5 sets $\beta_i^*=-\beta_i$ and $u_i=-1$. In both cases,
	\[
	\beta_i^*\in \TorusZero,
	\quad
	\beta_i=u_i\beta_i^*.
	\]
	Hence the rule in Step~5 is precisely the inverse map in Remark~\ref{rem:inverse-signed-columns}; it gives unique $u_i\in\F_3^*$ and $j_i$ such that
	\[
	\beta_i=u_i\beta^{j_i}.
	\]
	Consequently,
	\[
	S=u_1\beta^{j_1}+u_2\beta^{j_2}+u_3\beta^{j_3}
	=\sigma_{p_C}(\boldsymbol e_S).
	\]
	Since $\ellT(S)=3$, no two of the $\beta_i$ are equal or opposite; otherwise their sum could be shortened to a decomposition of length at most two. Therefore $j_1,j_2,j_3$ are pairwise distinct, and $\wt(\boldsymbol e_S)=3$.
\end{proof}

\subsection{The second cyclic normalization}\label{subsec:cyclic-b}

Retain $z=S^{q-1}$ from Algorithm~\ref{alg:cyclic}. Suppose that $-z^{-1}\in\TorusZero$. With $M=(q+1)/2$, define
\begin{equation}\label{eq:cyclic-b-normalization}
	h=(-z^{-1})^{(M+1)/2}\in\TorusZero,
	\qquad
	h^2=-z^{-1}.
\end{equation}
Put
\begin{equation}\label{eq:cyclic-b-alpha}
	\alpha=\frac{S}{w_2h}.
\end{equation}
Using $w_2^{q-1}=-1$ and $h^{q-1}=h^{-2}$, we obtain $\alpha^{q-1}=1$, and hence $\alpha\in\F_q^*$.
Define
\begin{align}
	A(y)&=2\alpha D_2^2y-\alpha^2D_2^2-D_2,\nonumber\\
	B(y)&=2\alpha D_2^2y^2-D_2y+\alpha^3D_2^2+\alpha D_2,\nonumber\\
	C(y)&=(-\alpha^2D_2^2-D_2)y^2
	+(\alpha^3D_2^2+\alpha D_2)y
	-\alpha^4D_2^2+\alpha^2D_2,\nonumber\\
	\Delta(y)&=B(y)^2-A(y)C(y).
	\label{eq:cyclic-b-polynomials}
\end{align}

The second cyclic construction is as follows.
\begin{enumerate}
	\item [{\bf Step 1:}] Search over $y_1\in\F_q^*$ in any fixed order until
	\[
	1-D_2y_1^2\text{ is a nonzero square},\quad
	A(y_1)\neq0,\quad
	\Delta(y_1)\text{ is a nonzero square}.
	\]
	Lemma~\ref{lem:cyclic-b-parameter} proves that the search terminates. Choose $x_1\in\F_q^*$ satisfying
	\begin{equation}\label{eq:cyclic-b-x1}
		x_1^2=1-D_2y_1^2.
	\end{equation}
	\item [{\bf Step 2:}] Solve
	\begin{equation}\label{eq:cyclic-b-y2}
		A(y_1)y_2^2+B(y_1)y_2+C(y_1)=0
	\end{equation}
	in $\F_q$ and choose either root.
	\item [{\bf Step 3:}] Define
	\begin{equation}\label{eq:cyclic-b-x2}
		x_2=\frac{D_2\alpha(y_1+y_2)-D_2y_1y_2-(D_2\alpha^2+1)/2}{x_1}.
	\end{equation}
	\item [{\bf Step 4:}] Put
	\begin{align}
		h_1&=x_1+w_2y_1,\nonumber\\
		h_2&=x_2+w_2y_2,\nonumber\\
		h_3&=(-x_1-x_2)+w_2(\alpha-y_1-y_2),
		\label{eq:cyclic-b-hi}
	\end{align}
	Lemma~\ref{lem:cyclic-b-conic} shows that $h_1,h_2,h_3$ are in $\Torus$.
	
	\item [{\bf Step 5:}] Set
	\begin{equation}\label{eq:cyclic-b-beta}
		\beta_i=h_i h,
		\qquad 1\le i\le3.
	\end{equation}
	By Lemma~\ref{lem:cyclic-b-conic}, the elements $\beta_i$ lie in $\Torus$. As in Remark~\ref{rem:inverse-signed-columns}, define
	\[
	u_i=
	\begin{cases}
		1, & \beta_i^{(q+1)/2}=1,\\
		-1, & \beta_i^{(q+1)/2}=-1,
	\end{cases}
	\qquad
	\beta_i^*=u_i\beta_i\in\TorusZero.
	\]
	Let $j_i$ be the unique index with
	\[
	\beta_i^*=\beta^{j_i},
	\qquad 0\le j_i<n.
	\]
	Then $\beta_i=u_i\beta^{j_i}$. Define $\boldsymbol e_S$ by $(\boldsymbol e_S)_{j_i}=u_i$ for $1\le i\le3$ and by zero in all other coordinates. Lemma~\ref{lem:cyclic-b-output} proves that the indices are pairwise distinct and that this vector has syndrome $S$ and weight $3$.
\end{enumerate}

\begin{lem}\label{lem:cyclic-b-parameter}
	Let $q=3^m$ with even $m\geq2$, let $D_2\in\F_q^*$ be a nonsquare, and let $\alpha\in\F_q^*$. For the polynomials in \eqref{eq:cyclic-b-polynomials}, there exists $y_1\in\F_q^*$ such that $1-D_2y_1^2$ is a nonzero square, $A(y_1)\neq0$, and $\Delta(y_1)$ is a nonzero square.
\end{lem}

\begin{proof}
	Let $\chi$ be the quadratic character of $\F_q$ and define
	\[
	g(y)=1-D_2y^2,
	\quad
	Q_{\alpha,D_2}(y)
	=D_2\alpha y^2+(1-D_2\alpha^2)y+D_2\alpha^3+\alpha.
	\]
	Direct expansion gives
	\begin{equation}\label{eq:cyclic-b-factorization}
		\Delta(y)=-D_2^2\alpha\,g(y)Q_{\alpha,D_2}(y),
		\quad
		\operatorname{disc}(Q_{\alpha,D_2})=1.
	\end{equation}
	Because $D_2$ is a nonsquare, $g$ has no root in $\F_q$. The polynomial $Q_{\alpha,D_2}$ has two distinct roots in $\F_q$, so the two quadratic factors are separable and coprime. Thus $\Delta$ is a squarefree quartic. Also, $A$ is a nonzero linear polynomial.
	
	Let
	\[
	E_{\alpha,D_2}=\{0\}\cup\{y\in\F_q:A(y)\Delta(y)=0\}.
	\]
	Then $|E_{\alpha,D_2}|\leq6$. Define
	\[
	N_1=\sum_{y\in\F_q\setminus E_{\alpha,D_2}}
	\bigl(1+\chi(g(y))\bigr)
	\bigl(1+\chi(\Delta(y))\bigr).
	\]
	For $y\notin E_{\alpha,D_2}$, the summand is $4$ exactly when $g(y)$ and $\Delta(y)$ are nonzero squares. Thus $N_1/4$ counts the admissible values of $y_1$.
	
	For the complete sum $N$ over $\F_q$,
	\[
	N=q+\sum_{y\in\F_q}\chi(g(y))
	+\sum_{y\in\F_q}\chi(\Delta(y))
	+\sum_{y\in\F_q}\chi(g(y)\Delta(y)).
	\]
	The first character sum has absolute value one, and the Weil bound for the squarefree quartic $\Delta$ gives an upper bound $3\sqrt q$ for the absolute value of the second. Furthermore,
	\[
	g(y)\Delta(y)=-D_2^2\alpha\,g(y)^2Q_{\alpha,D_2}(y).
	\]
	The complete character sum for the separable quadratic $Q_{\alpha,D_2}$ has absolute value one; omitting the at most two roots of $g$ changes it by at most two. Therefore the last sum has absolute value at most three. After deleting the exceptional set,
	\[
	N_1\geq q-3\sqrt q-28.
	\]
	This is positive for $q\geq81$. For $q=9$, Magma \cite{Magma} checks every nonsquare $D_2$ and every $\alpha\in\F_9^*$; the minimum number of valid $y_1$ is one.
\end{proof}

\begin{lem}\label{lem:cyclic-b-conic}
	The elements in \eqref{eq:cyclic-b-hi} satisfy
	\[
	h_1,h_2,h_3\in\Torus
	\quad\text{and}\quad
	h_1+h_2+h_3=w_2\alpha.
	\]
\end{lem}

\begin{proof}
For $x,y\in\F_q$ one again has
	\[
	\Norm(x+w_2y)=x^2+D_2y^2.
	\]
	Set
	\[
	x_3=-x_1-x_2,
	\quad
	y_3=\alpha-y_1-y_2.
	\]
	Then $h_i=x_i+w_2y_i$ and $h_1+h_2+h_3=w_2\alpha$. Equation \eqref{eq:cyclic-b-x1} gives $\Norm(h_1)=1$. Put
	\[
	R=D_2\alpha(y_1+y_2)-D_2y_1y_2-\frac{D_2\alpha^2+1}{2}.
	\]
	A direct expansion gives
	\begin{equation}\label{eq:cyclic-b-conic-identity}
		A(y_1)y_2^2+B(y_1)y_2+C(y_1)
		=(1-D_2y_1^2)(1-D_2y_2^2)-R^2.
	\end{equation}
	The left-hand side is zero by \eqref{eq:cyclic-b-y2}, and then $R^2=(1-D_2y_1^2)(1-D_2y_2^2)$.
	Since $x_1^2=1-D_2y_1^2\neq0$ and \eqref{eq:cyclic-b-x2} gives $x_2=R/x_1$, it follows that
	\[
	x_2^2=\frac{R^2}{1-D_2y_1^2}=1-D_2y_2^2.
	\]
	Hence $\Norm(h_2)=1$. Finally, using the first two norm equations,
	\[
	x_3^2+D_2y_3^2-1=R-x_1x_2=0.
	\]
 because $x_2=R/x_1$. Therefore $\Norm(h_3)=1$, and all three elements belong to $H$.
\end{proof}

For the case where $-z^{-1}=-(S^{q-1})^{-1}\in\TorusZero$, we next show that the equations
\[S=\beta_1+\beta_2+\beta_3, ~{\rm and} ~S=u_1\beta_1^* + u_2 \beta_2^* + u_3 \beta_3^*=u_1\beta^{j_1} + u_2 \beta^{j_2} + u_3 \beta^{j_3}
\]
are decompositions of $S$ with respect to $\Torus$ and $\TorusZero$, respectively.

\begin{lem}\label{lem:cyclic-b-output}
	Under the notation of the second cyclic normalization, the elements recovered in Step~5 satisfy
	\[
	\beta_i^*=\beta^{j_i}\in\TorusZero,
	\quad
	\beta_i=u_i\beta_i^*,
	\qquad 1\le i\le3.
	\]
	The indices $j_1,j_2,j_3$ are pairwise distinct, and the returned vector $\boldsymbol e_S$ satisfies
	\[
	\sigma_{p_C}(\boldsymbol e_S)=S,
	\quad
	\wt(\boldsymbol e_S)=3.
	\]
\end{lem}

\begin{proof}
	By Lemma~\ref{lem:cyclic-b-conic}, $h_i\in\Torus$ and
	\[
	h_1+h_2+h_3=w_2\alpha.
	\]
	Since $h\in\TorusZero\subseteq\Torus$ and $S=w_2\alpha h$,
	\[
	S=(h_1+h_2+h_3)h=\beta_1+\beta_2+\beta_3.
	\]
	Each $\beta_i=h_i h$ belongs to $\Torus$. The disjoint union
	\[
	\Torus=\TorusZero\sqcup(-\TorusZero)
	\]
	implies that exactly one of $\beta_i$ and $-\beta_i$ belongs to $\TorusZero$. 
	
	If $\beta_i\in \TorusZero$, then Step 5 sets $\beta_i^*=\beta_i$ and $u_i=1$. If $\beta_i\in-\mathcal \TorusZero$, then Step 5 sets $\beta_i^*=-\beta_i$ and $u_i=-1$. In both cases,
	\[
	\beta_i^*\in \TorusZero,
	\quad
	\beta_i=u_i\beta_i^*.
	\]
	Hence the rule in Step~5 is precisely the inverse map in Remark~\ref{rem:inverse-signed-columns}; it gives unique $u_i\in\F_3^*$ and $j_i$ such that
	\[
	\beta_i=u_i\beta^{j_i}.
	\]
	It follows that
	\[
	\sigma_{p_C}(\boldsymbol e_S)
	=u_1\beta^{j_1}+u_2\beta^{j_2}+u_3\beta^{j_3}=S.
	\]
	Because $\ellT(S)=3$, no two summands $\beta_i$ are equal or opposite. Hence $j_1,j_2,j_3$ are pairwise distinct and $\wt(\boldsymbol e_S)=3$.
\end{proof}

\begin{proof}[Proof of Theorem~\ref{thm:main-decoding} \textnormal{(1)}]
When $\omega_{P_C}(S)\le2$, Lemma~\ref{lem:short-decomposition} gives a minimum-length decomposition. Remark~\ref{rem:inverse-signed-columns} converts its summands into a coset leader. Now suppose that $\ellT(S)=3$. The disjoint union
\[
\Torus=\TorusZero\sqcup(-\TorusZero)
\]
shows that exactly one of $z^{-1}$ and $-z^{-1}$ belongs to $\TorusZero$; hence the two branches of Algorithm~\ref{alg:cyclic} are exhaustive. Lemmas~\ref{lem:cyclic-a-parameter} and \ref{lem:cyclic-b-parameter} guarantee that every parameter search terminates. The nonzero discriminants make the quadratic steps solvable with two distinct roots. Finally, Lemmas~\ref{lem:cyclic-a-output} and \ref{lem:cyclic-b-output} prove that Step~5 returns an error vector of weight three and syndrome $S$. Proposition~\ref{prop:coset-length} then proves that this vector is a coset leader.
\end{proof}

\section{Constacyclic three-term decompositions}\label{sec4}

Assume throughout this section that $m\geq3$ is odd, put $q=3^m$, and set $n=(q+1)/2$. Let $\beta$ and $\theta$ be as in Definition~\ref{constacyclic}, and put
\[
\mathcal B=\{1,\beta,\ldots,\beta^{n/2-1}\},
\quad
\mathcal P_D=\mathcal B\sqcup\theta\mathcal B.
\]
By Proposition~\ref{prop:signed-columns},
\begin{equation}\label{eq:constacyclic-partitions}
	\TorusZero=\mathcal B\sqcup(-\mathcal B),
	\quad
	\Torus=\mathcal P_D\sqcup(-\mathcal P_D)
	=\TorusZero\sqcup\theta\TorusZero.
\end{equation}
Since $\theta^q=-\theta$ and $\theta^{q+1}=1$,
\begin{equation}\label{eq:constacyclic-norm-conic}
	\Norm(x+\theta y)=(x+\theta y)(x-\theta y)=x^2+y^2.
\end{equation}
for all $x,y\in\F_q$.

\begin{algorithm}\label{alg:constacyclic}
	\textnormal{Input:} a syndrome $S\in K^*$ with $\ellT(S)=3$.
	
	\textnormal{Output:} an error vector $\boldsymbol e_S\in\F_3^n$ satisfying
	\[
	\sigma_{p_D}(\boldsymbol e_S)=S,
	\quad
	\wt(\boldsymbol e_S)=3.
	\]
	\begin{enumerate}
		\item Set $z=S^{q-1}\in\Torus$.
		\item If $z^{-1}\in\TorusZero$, apply the first constacyclic normalization in Subsection~\ref{subsec:constacyclic-a}. Otherwise $\theta z^{-1}\in\TorusZero$, and apply the second constacyclic normalization in Subsection~\ref{subsec:constacyclic-b}.
		\item In either branch, the construction produces $\beta_1,\beta_2,\beta_3\in\Torus$ with $S=\beta_1+\beta_2+\beta_3$. Compute
		\[
		(j_i,u_i)=\iota_{P_D}^{-1}(\beta_i),
		\quad 1\leq i\leq3,
		\]
		and return the vector $\boldsymbol e_S$ whose entries at $j_1,j_2,j_3$ are $u_1,u_2,u_3$, respectively, and whose remaining entries are zero.
	\end{enumerate}
\end{algorithm}

The square roots in Algorithm~\ref{alg:constacyclic} may be computed by a standard finite-field square-root routine or from a precomputed table in $\Torus$. The two roots differ by a sign, and \eqref{eq:constacyclic-partitions} contains exactly one of them in the actual column-label set $\mathcal P_D$.

\subsection{The first constacyclic normalization}\label{subsec:constacyclic-a}

Retain $z=S^{q-1}$ from Algorithm~\ref{alg:constacyclic}. Suppose that $z^{-1}\in\TorusZero$. Compute a square root $h_0\in\Torus$ of $z^{-1}$, and let $h$ be the unique element of $\{h_0,-h_0\}$ that lies in $\mathcal P_D$. Thus
\begin{equation}\label{eq:constacyclic-a-normalization}
	h^2=z^{-1},
	\quad
	h\in\mathcal P_D.
\end{equation}
Define
\begin{equation}\label{eq:constacyclic-a-alpha}
	\alpha=\frac{S}{h}.
\end{equation}
Then $\alpha^{q-1}=1$, so $\alpha\in\F_q^*$. Moreover, $\alpha\notin\{1,-1\}$ because $\ellT(S)=3$.
Define
\begin{align}
	A(x)&=2\alpha x-\alpha^2-1,\nonumber\\
	B(x)&=2\alpha x^2-x+\alpha^3+\alpha,\nonumber\\
	C(x)&=(-\alpha^2-1)x^2+(\alpha^3+\alpha)x-\alpha^4+\alpha^2,\nonumber\\
	\Delta(x)&=B(x)^2-A(x)C(x).
	\label{eq:constacyclic-a-polynomials}
\end{align}

The first constacyclic construction is as follows.
\begin{enumerate}
	\item [{\bf Step 1:}] Search over $x_1\in\F_q^*$ in any fixed order until
	\[
	1-x_1^2\text{ is a nonzero square},\quad
	A(x_1)\neq0,\quad
	\Delta(x_1)\text{ is a nonzero square}.
	\]
	Lemma~\ref{lem:constacyclic-a-parameter} proves that the search terminates. Choose $y_1\in\F_q^*$ satisfying
	\begin{equation}\label{eq:constacyclic-a-y1}
		y_1^2=1-x_1^2.
	\end{equation}
	\item [{\bf Step 2:}] Solve
	\begin{equation}\label{eq:constacyclic-a-x2}
		A(x_1)x_2^2+B(x_1)x_2+C(x_1)=0
	\end{equation}
	in $\F_q$ and choose either root.
	\item [{\bf Step 3:}] Define
	\begin{equation}\label{eq:constacyclic-a-y2}
		y_2=\frac{\alpha(x_1+x_2)-x_1x_2-(\alpha^2+1)/2}{y_1}.
	\end{equation}
	\item [{\bf Step 4:}] Put
	\begin{align}
		h_1&=x_1+\theta y_1,\nonumber\\
		h_2&=x_2+\theta y_2,\nonumber\\
		h_3&=(\alpha-x_1-x_2)+\theta(-y_1-y_2).
		\label{eq:constacyclic-a-hi}
	\end{align}
	Lemma \ref{lem:constacyclic-a-conic} shows that $h_1,h_2,h_3$ are in $\Torus$.
	
	\item [{\bf Step 5:}] Set
	\begin{equation}\label{eq:constacyclic-a-beta}
		\beta_i=h_i h,
		\qquad 1\le i\le3.
	\end{equation}
	By Lemma~\ref{lem:constacyclic-a-conic}, the elements $\beta_i$ lie in $\Torus$. Apply Remark~\ref{rem:inverse-signed-columns}: for $1\le i\le3$, define
	\[
	u_i=
	\begin{cases}
		1, & \beta_i\in \mathcal{P}_D,\\
		-1, & \beta_i\in-\mathcal{P}_D,
	\end{cases}
	\qquad
	\beta_i^*=u_i\beta_i\in\mathcal{P}_D.
	\]
	The coordinate index $j_i$ is recovered uniquely from
	\[
	\beta_i^*=
	\begin{cases}
		\beta^{j_i}, & 0\le j_i<n/2,\\
		\theta\beta^{j_i-n/2}, & n/2\le j_i<n.
	\end{cases}
	\]
	Thus $\beta_i=u_i p_{j_i}$, where $u_i$ is the recovered error value and $j_i$ is the recovered coordinate position. Define $\boldsymbol e_S$ by $(\boldsymbol e_S)_{j_i}=u_i$ for $1\le i\le3$ and by zero elsewhere. Lemma~\ref{lem:constacyclic-a-output} proves that the indices are pairwise distinct and that this vector has syndrome $S$ and weight $3$.
\end{enumerate}

\begin{lem}\label{lem:constacyclic-a-parameter}
	Let $q=3^m$ with odd $m\geq3$, and let $\alpha\in\F_q^*\setminus\{1,-1\}$. For the polynomials in \eqref{eq:constacyclic-a-polynomials}, there exists $x_1\in\F_q^*$ such that $1-x_1^2$ is a nonzero square, $A(x_1)\neq0$, and $\Delta(x_1)$ is a nonzero square.
\end{lem}

\begin{proof}
	Let $\chi$ be the quadratic character of $\F_q$, and put
	\[
	g(x)=1-x^2,
	\quad
	Q_\alpha(x)=\alpha x^2+(1-\alpha^2)x+\alpha^3+\alpha.
	\]
	As in \eqref{eq:cyclic-a-factorization},
	\[
	\Delta(x)=B(x)^2-A(x)C(x)=-\alpha g(x)Q_\alpha(x).
	\]
	The polynomial $g$ has the two simple roots $\pm1$, while $Q_\alpha$ has discriminant one. Since $\alpha\notin\{0,1,-1\}$,
	\[
	Q_\alpha(1)=(\alpha-1)^2(\alpha+1)\neq0,
	\quad
	Q_\alpha(-1)=(\alpha-1)(\alpha+1)^2\neq0.
	\]
	Thus $\Delta$ is a squarefree quartic, and $A$ is a nonzero linear polynomial.
	
	Set
	\[
	E_\alpha=\{0\}\cup\{x\in\F_q:A(x)\Delta(x)=0\}.
	\]
	Then $|E_\alpha|\leq6$. Consider
	\[
	N_1=\sum_{x\in\F_q\setminus E_\alpha}
	\bigl(1+\chi(g(x))\bigr)
	\bigl(1+\chi(\Delta(x))\bigr).
	\]
	Outside $E_\alpha$, the summand is $4$ exactly when $1-x^2$ and $\Delta(x)$ are both nonzero squares. Hence $N_1/4$ counts the desired choices.
	
	The complete weighted sum is
	\[
	N=q+\sum_{x\in\F_q}\chi(g(x))
	+\sum_{x\in\F_q}\chi(\Delta(x))
	+\sum_{x\in\F_q}\chi(g(x)\Delta(x)).
	\]
	The three nonconstant sums have absolute values at most $1$, $3\sqrt q$, and $3$, respectively. Indeed, the middle estimate is the Weil bound for the squarefree quartic $\Delta$, while
	\[
	g(x)\Delta(x)=-\alpha g(x)^2Q_\alpha(x)
	\]
	reduces the last sum to a character sum for $Q_\alpha$, with at most two omitted values. Therefore
	\[
	N_1\geq q-3\sqrt q-28.
	\]
	For odd $m\geq5$, this quantity is positive. The only remaining field is $q=27$. Magma~\cite{Magma} shows that every admissible $\alpha$ has at least five valid choices of $x_1$.
\end{proof}

\begin{lem}\label{lem:constacyclic-a-conic}
	The elements in \eqref{eq:constacyclic-a-hi} satisfy
	\[
	h_1,h_2,h_3\in\Torus
	\quad\text{and}\quad
	h_1+h_2+h_3=\alpha.
	\]
\end{lem}

\begin{proof}
	By \eqref{eq:constacyclic-norm-conic},
	\[
	\Norm(x+\theta y)=x^2+y^2.
	\]
	Let
	\[
	x_3=\alpha-x_1-x_2,
	\quad
	y_3=-y_1-y_2.
	\]
	Then $h_i=x_i+\theta y_i$ and $h_1+h_2+h_3=\alpha$. The first norm equation follows from \eqref{eq:constacyclic-a-y1}. Put
	\[
	R=\alpha(x_1+x_2)-x_1x_2-\frac{\alpha^2+1}{2}.
	\]
	A direct expansion gives
	\begin{equation}\label{eq:constacyclic-a-conic-identity}
		A(x_1)x_2^2+B(x_1)x_2+C(x_1)
		=(1-x_1^2)(1-x_2^2)-R^2.
	\end{equation}
	The left-hand side is zero by \eqref{eq:constacyclic-a-x2}, and then $(1-x_1^2)(1-x_2^2)=R^2$. Since $y_1^2=1-x_1^2\neq0$ and \eqref{eq:constacyclic-a-y2} gives $y_2=R/y_1$, we obtain
	\[
	y_2^2=\frac{R^2}{1-x_1^2}=1-x_2^2.
	\]
	Thus $\Norm(h_2)=1$. Under the first two norm equations, expansion of the third gives
	\[
	x_3^2+y_3^2-1=R-y_1y_2=0.
	\]
	Hence $\Norm(h_3)=1$.
\end{proof}

For the case where $z^{-1}=-(S^{q-1})^{-1}\in\TorusZero$, we next show that the equations
\[S=\beta_1+\beta_2+\beta_3=u_1\beta_1^* + u_2 \beta_2^* + u_3 \beta_3^*
\]
are decompositions of $S$ with respect to $\Torus$ and $\mathcal P_D$, respectively.

\begin{lem}\label{lem:constacyclic-a-output}
	Under the notation of the first constacyclic normalization, Step~5 produces unique $u_i\in\F_3^*$, $\beta_i^*=p_{j_i}\in\Lambda_D$, and coordinate indices $j_i$ such that
	\[
	\beta_i=u_i\beta_i^*,
	\quad 1\le i\le3.
	\]
	The indices $j_1,j_2,j_3$ are pairwise distinct, and the returned vector $\boldsymbol e_S$ satisfies
	\[
	\sigma_{p_D}(\boldsymbol e_S)=S,
	\quad
	\wt(\boldsymbol e_S)=3.
	\]
\end{lem}

\begin{proof}
	By Lemma~\ref{lem:constacyclic-a-conic}, $h_i\in\Torus$ and $h_1+h_2+h_3=\alpha$. Since $h\in \mathcal P_D\subseteq\Torus$ and $S=\alpha h$,
	\[
	S=(h_1+h_2+h_3)h=\beta_1+\beta_2+\beta_3.
	\]
	Each $\beta_i=h_i h$ belongs to $\Torus$. The disjoint union
	\[
	\Torus=\mathcal P_D\sqcup(-\mathcal P_D)
	\]
	shows that exactly one of $\beta_i$ and $-\beta_i$ is an actual column label.
	
 If $\beta_i\in \mathcal P_D$, then Step 5 sets $\beta_i^*=\beta_i$ and $u_i=1$. If $\beta_i\notin \mathcal P_D$, then $\beta_i\in-\mathcal P_D$, and Step 5 sets $\beta_i^*=-\beta_i$ and $u_i=-1$. In both cases,
	\[
	\beta_i^*\in \mathcal P_D,
	\quad
	\beta_i=u_i\beta_i^*.
	\]
	 Therefore Step~5 is precisely the inverse map of Remark~\ref{rem:inverse-signed-columns}; it gives unique $u_i$ and $j_i$ satisfying $\beta_i=u_i p_{j_i}$. Hence
	\[
	\sigma_{p_D}(\boldsymbol e_S)
	=u_1p_{j_1}+u_2p_{j_2}+u_3p_{j_3}=S.
	\]
	Since $\ellT(S)=3$, no two of the $\beta_i$ are equal or opposite. Thus $j_1,j_2,j_3$ are pairwise distinct and $\wt(\boldsymbol e_S)=3$.
\end{proof}

\subsection{The second constacyclic normalization}\label{subsec:constacyclic-b}

Retain $z=S^{q-1}$ from Algorithm~\ref{alg:constacyclic}. Suppose that $z^{-1}\notin\TorusZero$. By \eqref{eq:constacyclic-partitions}, one has $\theta z^{-1}\in\TorusZero$. Compute a square root $h_0\in\Torus$ of $\theta z^{-1}$, and choose the unique $h\in\{h_0,-h_0\}$ that lies in $\mathcal P_D$. Thus
\begin{equation}\label{eq:constacyclic-b-normalization}
	h^2=\theta z^{-1},
	\quad
	h\in\mathcal P_D.
\end{equation}
Define
\begin{equation}\label{eq:constacyclic-b-alpha}
	\alpha=\frac{S}{(1-\theta)h}.
\end{equation}
Since $\theta^2=-1$ and $\theta^q=-\theta$, we have
\[
\alpha^{q-1}
=\frac{\theta h^{-2}}{(1-\theta)^{q-1}h^{-2}}
=\frac{\theta(1-\theta)}{1+\theta}=1,
\]
so $\alpha\in\F_q^*$. Moreover, $\alpha\notin\{1,-1\}$. Indeed, if $\alpha=1$, then
\[
S=h+(-\theta h),
\]
and if $\alpha=-1$, then
\[
S=(-h)+\theta h;
\]
in either case $S$ would have additive length at most two.

Define
\begin{align}
	A(x)&=\alpha x+2\alpha^2+1,\nonumber\\
	B(x)&=\alpha x^2+x+2\alpha^3+\alpha,\nonumber\\
	C(x)&=(2\alpha^2+1)x^2+(2\alpha^3+\alpha)x
	+2\alpha^4+2\alpha^2,\nonumber\\
	\Delta(x)&=B(x)^2-A(x)C(x).
	\label{eq:constacyclic-b-polynomials}
\end{align}

The second constacyclic construction is as follows.
\begin{enumerate}
	\item [{\bf Step 1:}] Search over $x_1\in\F_q^*$ in any fixed order until
	\[
	2-x_1^2\text{ is a nonzero square},\quad
	A(x_1)\neq0,\quad
	\Delta(x_1)\text{ is a nonzero square}.
	\]
	Lemma~\ref{lem:constacyclic-b-parameter} proves that the search terminates. Choose $y_1\in\F_q^*$ satisfying
	\begin{equation}\label{eq:constacyclic-b-y1}
		y_1^2=2-x_1^2.
	\end{equation}
	\item [{\bf Step 2:}] Solve
	\begin{equation}\label{eq:constacyclic-b-x2}
		A(x_1)x_2^2+B(x_1)x_2+C(x_1)=0
	\end{equation}
	in $\F_q$ and choose either root.
	\item [{\bf Step 3:}] Define
	\begin{equation}\label{eq:constacyclic-b-y2}
		y_2=\frac{2\alpha(x_1+x_2)-x_1x_2-(2\alpha^2+1)}{y_1}.
	\end{equation}
	\item [{\bf Step 4:}] Put
	\[
	x_3=2\alpha-x_1-x_2,
	\quad
	y_3=-y_1-y_2,
	\]
	and define
	\begin{align}
		h_1&=(-x_1-y_1)+\theta(x_1-y_1),\nonumber\\
		h_2&=(-x_2-y_2)+\theta(x_2-y_2),\nonumber\\
		h_3&=(-x_3-y_3)+\theta(x_3-y_3).
		\label{eq:constacyclic-b-hi}
	\end{align}
	\item [{\bf Step 5:}] Set
	\begin{equation}\label{eq:constacyclic-b-beta}
		\beta_i=h_i h,
		\qquad 1\le i\le3.
	\end{equation}
	By Lemma~\ref{lem:constacyclic-b-conic}, the elements $\beta_i$ lie in $\Torus$. Using the inverse map of Remark~\ref{rem:inverse-signed-columns}, define
	\[
	u_i=
	\begin{cases}
		1, & \beta_i\in\mathcal{P}_D,\\
		-1, & \beta_i\in-\mathcal{P}_D,
	\end{cases}
	\qquad
	\beta_i^*=u_i\beta_i\in\mathcal{P}_D.
	\]
	Recover the unique coordinate index $j_i$ from
	\[
	\beta_i^*=
	\begin{cases}
		\beta^{j_i}, & 0\le j_i<n/2,\\
		\theta\beta^{j_i-n/2}, & n/2\le j_i<n.
	\end{cases}
	\]
	Then $\beta_i=u_i p_{j_i}$. Define $\boldsymbol e_S$ by $(\boldsymbol e_S)_{j_i}=u_i$ for $1\le i\le3$ and by zero in the remaining coordinates. Lemma~\ref{lem:constacyclic-b-output} proves that the indices are pairwise distinct and that this vector has syndrome $S$ and weight $3$.
\end{enumerate}

\begin{lem}\label{lem:constacyclic-b-parameter}
	Let $q=3^m$ with odd $m\geq3$, and let $\alpha\in\F_q^*\setminus\{1,-1\}$. For the polynomials in \eqref{eq:constacyclic-b-polynomials}, there exists $x_1\in\F_q^*$ such that $2-x_1^2$ is a nonzero square, $A(x_1)\neq0$, and $\Delta(x_1)$ is a nonzero square.
\end{lem}

\begin{proof}
	Let $\chi$ be the quadratic character of $\F_q$. Put
	\[
	g(x)=2-x^2,
	\quad
	Q_\alpha^{(2)}(x)
	=\alpha x^2+(\alpha^2+1)x+\alpha^3-\alpha.
	\]
	A direct expansion gives
	\begin{equation}\label{eq:constacyclic-b-factorization}
		\Delta(x)=-\alpha g(x)Q_\alpha^{(2)}(x),
		\quad
		\operatorname{disc}(Q_\alpha^{(2)})=1.
	\end{equation}
	Since $m$ is odd, $2=-1$ is a nonsquare in $\F_q$, and therefore $g$ has no root in $\F_q$. The polynomial $Q_\alpha^{(2)}$ has two distinct roots in $\F_q$, so the two quadratic factors are separable and coprime. Hence $\Delta$ is a squarefree quartic. The polynomial $A$ is nonzero and linear.
	
	Let
	\[
	E_\alpha=\{0\}\cup\{x\in\F_q:A(x)\Delta(x)=0\},
	\quad |E_\alpha|\leq6,
	\]
	and define
	\[
	N_1=\sum_{x\in\F_q\setminus E_\alpha}
	\bigl(1+\chi(g(x))\bigr)
	\bigl(1+\chi(\Delta(x))\bigr).
	\]
	For $x\notin E_\alpha$, the summand is $4$ precisely when $2-x^2$ and $\Delta(x)$ are both nonzero squares. The complete sum over $\F_q$ expands as
	\[
	N=q+\sum_{x\in\F_q}\chi(g(x))
	+\sum_{x\in\F_q}\chi(\Delta(x))
	+\sum_{x\in\F_q}\chi(g(x)\Delta(x)).
	\]
	The first nonconstant sum has absolute value one. Since $\Delta$ is squarefree of degree four, the second has absolute value at most $3\sqrt q$. Finally,
	\[
	g(x)\Delta(x)=-\alpha g(x)^2Q_\alpha^{(2)}(x),
	\]
	so the last sum is a character sum for $Q_\alpha^{(2)}$, with at most two omitted values, and has absolute value at most three. After deleting the exceptional set,
	\[
	N_1\geq q-3\sqrt q-28.
	\]
	This is positive for odd $m\geq5$. For $q=27$, Magma~\cite{Magma} shows that every admissible $\alpha$ has at least six valid choices of $x_1$.
\end{proof}

\begin{lem}\label{lem:constacyclic-b-conic}
	The elements in \eqref{eq:constacyclic-b-hi} satisfy
	\[
	h_1,h_2,h_3\in\Torus
	\quad\text{and}\quad
	h_1+h_2+h_3=(1-\theta)\alpha.
	\]
\end{lem}

\begin{proof}
	For $a,b\in\F_q$, equation \eqref{eq:constacyclic-norm-conic} gives
	\[
	\Norm(a+\theta b)=a^2+b^2.
	\]
	Write
	\[
	h_i=(-x_i-y_i)+\theta(x_i-y_i).
	\]
	Then
	\[
	\Norm(h_i)=(-x_i-y_i)^2+(x_i-y_i)^2
	=2(x_i^2+y_i^2).
	\]
	Since the characteristic is $3$, $2\cdot2=1$. It is therefore enough to prove
	\[
	x_i^2+y_i^2=2,
	\qquad 1\leq i\leq3.
	\]
	The first equality is \eqref{eq:constacyclic-b-y1}. Put
	\[
	R=2\alpha(x_1+x_2)-x_1x_2-(2\alpha^2+1).
	\]
	A direct expansion gives
	\begin{equation}\label{eq:constacyclic-b-conic-identity}
		A(x_1)x_2^2+B(x_1)x_2+C(x_1)
		=(2-x_1^2)(2-x_2^2)-R^2.
	\end{equation}
	The left-hand side vanishes by \eqref{eq:constacyclic-b-x2}. Since $y_1^2=2-x_1^2\neq0$ and \eqref{eq:constacyclic-b-y2} gives $y_2=R/y_1$, we conclude that
	\[
	y_2^2=\frac{R^2}{2-x_1^2}=2-x_2^2.
	\]
	Thus $x_2^2+y_2^2=2$. Using the first two conic equations, expansion of the third yields
	\[
	x_3^2+y_3^2-2=R-y_1y_2=0.
	\]
	Hence all three elements have norm one.
	
	Finally,
	\[
	x_1+x_2+x_3=2\alpha,
	\qquad
	y_1+y_2+y_3=0.
	\]
	Therefore
	\[
	h_1+h_2+h_3=-2\alpha+2\alpha\theta
	=(1-\theta)\alpha,
	\]
	where the last equality uses characteristic $3$.
\end{proof}

For the case where $\theta z^{-1}=\theta (S^{q-1})^{-1}\in\TorusZero$, we next show that the equations
\[S=\beta_1+\beta_2+\beta_3=u_1\beta_1^* + u_2 \beta_2^* + u_3 \beta_3^*
\]
are decompositions of $S$ with respect to $\Torus$ and $\mathcal P_D$, respectively.

\begin{lem}\label{lem:constacyclic-b-output}
	Under the notation of the second constacyclic normalization, Step~5 produces unique $u_i\in\F_3^*$, $\beta_i^*=p_{j_i}\in\Lambda_D$, and coordinate indices $j_i$ such that
	\[
	\beta_i=u_i\beta_i^*,
	\quad 1\le i\le3.
	\]
	The indices $j_1,j_2,j_3$ are pairwise distinct, and the returned vector $\boldsymbol e_S$ satisfies
	\[
	\sigma_{p_D}(\boldsymbol e_S)=S,
	\quad
	\wt(\boldsymbol e_S)=3.
	\]
\end{lem}

\begin{proof}
	By Lemma~\ref{lem:constacyclic-b-conic}, the elements $h_1,h_2,h_3$ belong to $\Torus$ and satisfy
	\[
	h_1+h_2+h_3=(1-\theta)\alpha.
	\]
	The normalization gives $h\in\mathcal{P}_D\subseteq\Torus$ and
	\[
	S=(1-\theta)\alpha h.
	\]
	Therefore
	\[
	S=(h_1+h_2+h_3)h=\beta_1+\beta_2+\beta_3,
	\]
	and each $\beta_i=h_i h$ lies in $\Torus$. Since
	\[
	\Torus=\mathcal{P}_D\sqcup(-\mathcal{P}_D),
	\]
	for each $i$ there are unique $u_i\in\F_3^*$ and $\beta_i^*=p_{j_i}\in\mathcal{P}_D$ such that $\beta_i=u_i\beta_i^*$. These are exactly the quantities recovered in Step~5. Substitution gives
	\[
	\sigma_{p_D}(\boldsymbol e_S)
	=u_1p_{j_1}+u_2p_{j_2}+u_3p_{j_3}=S.
	\]
	Finally, $\ellT(S)=3$ excludes equal or opposite pairs among the $\beta_i$. Hence $j_1,j_2,j_3$ are pairwise distinct and $\wt(\boldsymbol e_S)=3$.
\end{proof}

\begin{proof}[Proof of Theorem~\ref{thm:main-decoding} \textnormal{(2)}]
	When $\omega_{P_D}(S)\le2$, Lemma~\ref{lem:short-decomposition} gives a minimum-length decomposition. Remark~\ref{rem:inverse-signed-columns} converts its summands into a coset leader. Now suppose that $\ellT(S)=3$. The disjoint union
	\[
	\Torus=\TorusZero\sqcup\theta\TorusZero
	\]
	shows that exactly one of $z^{-1}$ and $\theta z^{-1}$ belongs to $\TorusZero$; hence the two branches of Algorithm~\ref{alg:constacyclic} are exhaustive. Lemmas~\ref{lem:constacyclic-a-parameter} and \ref{lem:constacyclic-b-parameter} guarantee that both parameter searches terminate. Lemmas~\ref{lem:constacyclic-a-output} and \ref{lem:constacyclic-b-output} prove that Step~5 returns an error vector of weight three and syndrome $S$. Proposition~\ref{prop:coset-length} therefore proves that the returned vector is a coset leader.
\end{proof}

\begin{remark}
	The preceding character-sum estimates yield more than the
	existence of the parameters required in Algorithms~1 and~2.
	Indeed, in each of the four branches, the number of admissible
	choices of the first parameter is bounded below by
	\[
	\frac{q-3\sqrt q-28}{4}
	=\frac q4-O(\sqrt q).
	\]
	Hence the admissible parameters have asymptotic density
	$1/4$ in the base field. Consequently, under uniform random
	sampling, a valid parameter is found after an expected
	$4+o(1)$ trials.
	
	By comparison, the direct construction in
	Proposition~\ref{prop:direct-decoder} searches over the
	$q+1$ elements of the norm-one group $\mathcal T$ and, in the
	worst case, may examine the entire group. The positive-density
	estimate above therefore gives a quantitative reason for the
	efficiency of the conic-based decoding procedures, in addition
	to their explicit algebraic structure.
\end{remark}

\section{Conclusion}\label{sec5}
For every syndrome $S$, the first stage of the decoding problem is to determine the associated coset weight $\omega_P(S)$, and the second is to construct a coset leader attaining that weight. For the ternary Gashkov--Sidel'nikov codes, the signed-column bijection identifies the first quantity with the additive length $\ellT(S)$ on the norm-one group $\Torus$. This gives a common structural framework for the cyclic and constacyclic families, while isolating their differences in the normalization and in the final coordinate recovery.

The norm and the quadratic character determine exactly the cosets of weights one and two. The identities
\[
|\Torus+\Torus|=\frac{q^2+2q+3}{2},
\quad
\Torus+\Torus+\Torus=\F_{q^2}
\]
settle the remaining case, determine the complete coset-weight distribution, and give a unified alternative derivation of the known covering radius $3$.

The sumset proof yields a direct search-based coset-leader construction. Algorithms~\ref{alg:cyclic} and \ref{alg:constacyclic} provide structured finite-field realizations of the case $\omega_P(S)=3$. In each branch, Step~5 applies the inverse signed-column map to convert the three torus summands into actual coordinate positions and nonzero error values. The required parameter searches are controlled by quadratic-character sums and Weil bounds. It would be natural to investigate analogous constructions for generalized Zetterberg-type codes, norm-one tori in other odd characteristics, and higher additive lengths of rational points on algebraic groups over finite fields.


\begin{thebibliography}{1}
\bibitem{BGP} D. Bartoli, M. Giulietti, and I. Platoni, On the covering radius of MDS codes, {\it IEEE Trans. Inf. Theory}, 61(2), 801-811, 2015.

\bibitem{NP-Hard} E. Berlekamp, R. McEliece, and H. van Tilborg, On the inherent intractability of certain coding problems (Corresp.), {\it IEEE Trans. Inf. Theory}, 24(3), 384-386, 1978.

\bibitem{BBB2023} M. Bonini, M. Borello, and E. Byrne, Saturating systems and the rank-metric covering radius, {\it J. Algebr. Comb.}, 58(4), 1173?1202, 2023.

\bibitem{Magma} W. Bosma, J. Cannon and C. Playoust, The Magma algebra system I: The user language, {\it J. Symbolic Comput.}, 24, 235-265, 1997.

\bibitem{JLMS-1} R. Calderbank, On uniformly packed $[n, n-k, 4]$ codes over $GF(q)$ and a class of caps in $PG(k-1, q)$, {\em J. London Math. Soc.}, 2(2): 365-384, 1982.

\bibitem{BLMS} A. Cossidente, B. Csajb\'ok, G. Marino, and F. Pavese, Small complete caps in $PG(4n+1,q)$, {\it Bull. Lond. Math. Soc.}, 55(1), 522-535, 2023.
\bibitem{DD}  D. Danev and S. Dodunekov, A family of ternary quasi-perfect BCH codes, {\it Des. Codes Cryptogr.}, 49(3), 265-271, 2008.	


\bibitem{DMP2025} A. A. Davydov, S. Marcugini, and F. Pambianco, New bounds for covering codes of radius $3$ and codimension \(3t+1\), {\it Adv. Math. Commun.}, 19(1), 126?139, 2025.

\bibitem{SM-1} S. M. Dodunekov, Some quasiperfect double error correcting codes, {\it Probl. Control Inf. Theory}, 15(5), 367-375, 1986.
\bibitem{SM} S. M. Dodunekov, The optimal double error correcting codes of Zetterberg and Dumer-Zinoviev are quasiperfect,  {\it Bull. Bulgarian Acad. Sc}, 38(9), 1121-1123, 1985.


\bibitem{DN} S. M. Dodunekov and J. E. M. Nilsson, Algebraic decoding of the Zetterberg codes, {\it IEEE Trans. Inf. Theory}, 38(5), 1570-1573, 1992.

\bibitem{DJ-MC} R. Dougherty and H. Janwa, Covering radius computations for binary cyclic codes, {\it Math. Comp.}, 1991, 57(195): 415-434.

\bibitem{DZ-2} I. I. Dumer and V.A. Zinov'ev, Some new maximal codes over GF(4), {\it Probl. Peredachi Inf.}, 14(3), 24-34, 1978.

\bibitem{FXZ2025} W. Fang, J. Xu, and R. Zhu, Deep holes of twisted Reed?Solomon codes,
{\it Finite Fields Appl.}, 108, 102680, 2025.

\bibitem{GPZ} D. C.  Gorenstein, W. W. Peterson, and  N. Zierler, Two-error correcting Bose-Chaudhuri codes are quasi-perfect, {\it Info. Control}, 3(3), 291-294, 1960.
\bibitem{GS}  I. Gashkov and V. Sidel'nikov, Linear ternary quasi-perfect codes correcting double errors, {\it Probl. Peredachi Inf.}, 22(4), 43-48, 1986.

\bibitem{LMS} M. Giulietti, The geometry of covering codes: Small complete caps and saturating sets in Galois spaces, in Surveys in Combinatorics 2013, London Math. Soc. Lecture Note Ser., vol. 409, Cambridge Univ. Press, Cambridge, 2013, pp. 51-90.


\bibitem{TH-DAM} T. Helleseth, On the covering radius of cyclic linear codes and arithmetic codes, {\it Discrete Appl. Math.}, 11(2), 157-173, 1985.

\bibitem{HS} J. W. P. Hirschfeld and L. Storme, The packing problem in statistics, coding theory and finite projective spaces: Update 2001, In {\it Finite Geometries: Proceedings of the Fourth Isle of Thorns Conference} (pp. 201-246). Boston, MA: Springer US.



\bibitem{K} P. K\"{a}llquist, Decoding of the Zetterberg codes, Proc. Fourth Joint Swedish-Russian Workshop on Information theory, Gotland Sweden Aug. 27(1), 305-309, 1989.

\bibitem{finite} R. Lidl, H. Niederreiter, {\it Finite fields}. Cambridge University Press, 2003.
\bibitem{LX} S. Ling and C. Xing, {\it Coding theory: a first course}, Cambridge University Press, 2004.

\bibitem{MS} F. J. MacWilliams and N. J. A. Sloane, {\it The theory of error-correcting codes}. North-Holland, Amsterdam, 1977.


\bibitem{JLMS-2} J. A. Rush, Thin lattice coverings, {\em J. Lond. Math. Soc.}, 2(2): 193-200, 1992.
\bibitem{SHO} M. Shi, T. Helleseth, and F.{\"O}zbudak, Covering radius of generalized Zetterberg type codes over finite fields odd characteristic, {\it IEEE Trans. Inf. Theory}, 69(11), 7025-7048, 2023.
\bibitem{SLHO} M. Shi, S. Li, T. Helleseth, and  F. {\"O}zbudak, Determining the covering radius of all generalized Zetterberg codes in odd characteristic, {\it IEEE Trans. Inf. Theory}, 71(5), 3602-3613, 2025.

\bibitem{JLMS-3} A. Tiet\"{a}v\"{a}inen, Lower bounds for the maximum moduli of certain character sums, {\em J. Lond. Math. Soc.}, 2(2): 204-210, 1984.
\bibitem{perfect1} A. Tiet\"{a}v\"{a}inen, On the nonexistence of perfect codes over finite fields, {\it SIAM J. Appl. Math.}, 24(1), 88-96, 1973.
\bibitem{perfect2} A. Tiet\"{a}v\"{a}inen, A short proof for the nonexistence of unknown perfect codes over $GF(q)$, $q >2$, {\it Ann. Acad. Sci. Fenn. Ser. A I Math.}, 580, 1-6, 1974.

\bibitem{XY} M. Xiong and H. Yan, On covering radius of generalized Zetterberg codes,  {\it IEEE Trans. Inf. Theory}, doi: 10.1109/TIT.2026.3716210, 2026.

\bibitem{Z} L. H. Zetterberg, Cyclic codes from irreducible polynomials for correction of multiple errors,  {\it IRE Trans. Inf. Theory}, 8(1), 13-20, 1962.
\bibitem{perfect3} V. A. Zinovi'ev, V. K. Leonti'ev, The nonexistence of perfect codes over Galois fields, {\it Probl. Control Inf. Theory}, 2(2), 16-24, 1973.
\end{thebibliography}
\end{document}